\documentclass[11pt, a4paper]{amsart}

\usepackage{xcolor}
\definecolor{MyBlue}{cmyk}{1,0.13,0,0.63}
\definecolor{MyGreen}{cmyk}{0.91,0,0.88,0.52}
\newcommand{\mylinkcolor}{MyBlue}
\newcommand{\mycitecolor}{MyGreen}
\newcommand{\myurlcolor}{black}

\usepackage{hyperref}
\hypersetup{%
  bookmarksnumbered=true,bookmarksopen=false,%
  plainpages=false,% necessary to prevent duplicate page identifiers
  linktocpage=true,%
  colorlinks=true,breaklinks=true,%
  linkcolor=\mylinkcolor,citecolor=\mycitecolor,urlcolor=\myurlcolor,%
  pdfpagelayout=OneColumn,%
  pageanchor=true,%
}

\usepackage{graphicx}
\usepackage{amsmath, amsthm,  amsfonts, amscd}
\usepackage{amssymb}
\usepackage{url}
\usepackage{fullpage}
\usepackage{mathtools}
\usepackage[all]{xy}
\usepackage{slashed}
\usepackage{multirow}

\newtheorem{thm}{Theorem}[section]
\newtheorem*{thm*}{Theorem}
\newtheorem{cor}[thm]{Corollary}
\newtheorem{lemma}[thm]{Lemma}
\newtheorem{prop}[thm]{Proposition}

\theoremstyle{definition}
\newtheorem{defn}[thm]{Definition}
\newtheorem{assumption}[thm]{Assumption}
\theoremstyle{remark}
\newtheorem{remark}[thm]{Remark}

\newtheorem{example}[thm]{Example}
\newtheorem{remarks}[thm]{Remarks}
\newtheorem{examples}[thm]{Examples}
\newcommand{\End}{\ensuremath{\mathrm{End}}}

\newcommand{\wt}{\ensuremath{\widetilde}}

\newcommand{\R}{\ensuremath{\mathbb{R}}}
\newcommand{\N}{\ensuremath{\mathbb{N}}}
\newcommand{\Z}{\ensuremath{\mathbb{Z}}}

\newcommand{\C}{\ensuremath{\mathbb{C}}}

\newcommand{\fraka}{{\mathfrak a}}

\newcommand{\frakc}{{\mathfrak c}}

\newcommand{\frakr}{{\mathfrak r}}

\def\calC{\mathcal{C}}

\def\calO{\mathcal{O}}

\def\calB{\mathcal{B}}
\def\calH{\mathcal{H}}

\def\calA{\mathcal{A}}

\def\calM{\mathcal{M}}

\def\calQ{\mathcal{Q}}
\def\calV{\mathcal{V}}
\def\calW{\mathcal{W}}
\def\calU{\mathcal{U}}

\newcommand{\one}{{\bf 1}}
\newcommand{\ol}{\overline}
\DeclareMathOperator{\Aut}{Aut}

\theoremstyle{definition}

\DeclareMathOperator{\Dom}{Dom}

\DeclareMathOperator{\Ker}{Ker}

\DeclareMathOperator{\Ran}{Ran}

\newcommand{\Cl}{\mathbb{C}\ell}

\newcommand{\rs}{{\mathfrak r}}

\newcommand{\hox}{\,\hat{\otimes}\,}

\newcommand{\rst}[1]{\ensuremath{{\mathbin\upharpoonright}%
\raise-.5ex\hbox{$#1$}}}

\makeatletter

\newcommand{\Rmnum}[1]{\expandafter\@sl217--242owromancap\romannumeral #1@}
\makeatother

\author{Chris Bourne}
\address{WPI-AIMR, Tohoku University,
2-1-1 Katahira, Aoba-ku, Sendai, 980-8577, Japan \emph{and} 
RIKEN iTHEMS, 2-1 Hirosawa, Wako, Saitama 351-0198, Japan}
\email{chris.bourne@tohoku.ac.jp}

\date{\today}

\begin{document}

\title{Locally equivalent quasifree states and index theory}

\begin{abstract}
We consider quasifree ground states of Araki's self-dual CAR algebra from the viewpoint of index theory and
symmetry protected topological (SPT) phases. We first review how Clifford module indices characterise a topological 
obstruction to connect pairs of symmetric gapped ground states. This construction is then generalised to 
give invariants in  $KO_\ast(A^\mathfrak{r})$ with $A$ a $C^{*,\mathfrak{r}}$-algebra of allowed deformations. When $A=C^*(X)$, the Roe algebra 
of a coarse space $X$, and we restrict to gapped ground states that are locally equivalent with respect $X$, a 
$K$-homology class is also constructed. The coarse assembly map relates these two classes and clarifies the 
relevance of $K$-homology to free-fermionic SPT phases.
\end{abstract}

\maketitle

\parskip=0.02in

\section{Introduction}

Since the influential paper of Kitaev~\cite{Kitaev09}, $K$-theory of spaces and $C^*$-algebras has played an 
important role in studying the phase labels of free-fermionic topological states of matter, 
see~\cite{FM13, Thiang14, Kubota17, Kellendonk15, PSBbook, AMZ} for example. The dual theory, $K$-homology, 
also features prominently in Kitaev's paper as a way to characterise gapped local systems. 
While index pairings with Dirac operators and $K$-homology classes constructed on the (noncommutative) Brillouin torus 
have been effectively utilised to give numerical topological phase 
labels~\cite{GSB16, BCR2}, the role of $K$-homology as a means to directly characterise local gapped systems appears to be 
understudied in the mathematical physics literature. The aim of this paper is to provide some first steps 
in this direction. 

For our purposes, it is most convenient to study free-fermionic topological phases via the dynamics induced 
by gapped Bogoliubov--de Gennes (BdG) Hamiltonians on a Nambu space, a complex Hilbert space with 
real structure, see~\cite{KZ16} for example. 
Such dynamics give a quasifree, gapped and pure ground state of the self-dual algebra of 
canonical anti-commutation relations (CAR) studied by 
Araki~\cite{Araki70}. 
Similar to work by Alldridge--Max--Zirnbauer~\cite{AMZ}, we construct elements in the $K$-theory 
of a Real $C^*$-algebra $A$ of allowed deformations that characterise pairs of gapped 
BdG Hamiltonians/quasifree ground states.

To relate these constructions to $K$-homology, we consider the case $A=C^*(X)$, the Roe algebra of a coarse space $X$ 
constructed from a representation of $C_0(X)$ on the Nambu space and 
with real structure $\frakr$.
Taking inspiration from similar constructions in algebraic quantum field theory~\cite{LudersRoberts, DAntoniHollands}, 
a notion of local equivalence of gapped 
quasifree ground states is introduced for BdG Hamiltonians that are compatible with the representation of 
$C_0(X)$.
We show that such locally equivalent 
ground states give rise to a Fredholm module and $K$-homology class for $C_0(X)$. The coarse assembly map 
$\mu_X:KO^{-\ast}(C_0(X)^\frakr) \to KO_\ast( C^*(X)^\frakr )$ sends this $K$-homology class 
to the previously constructed $KO_\ast(C^*(X)^\frakr)$-valued indices. 
The equivariant assembly map can similarly be treated for quasifree gapped ground states with a compact or 
discrete group symmetry.

The coarse assembly 
map is an isomorphism for a large class of spaces. Therefore our result helps establish 
the relevance of $K$-homology as a mathematical characterisation of free-fermionic 
topological phases as well as its relation with the more well-studied approach 
via $K$-theory.

\subsection*{Mathematical results}

Given a complex Hilbert space $\calH$ with real structure $\Gamma$, 
pure quasifree states of the self-dual CAR algebra $A^\mathrm{car}_\mathrm{sd}(\calH,\Gamma)$ 
can be characterised by  skew-adjoint 
unitaries on $\calH$ that commute with $\Gamma$. Given a pair  $(J_0, J_1)$ of such 
unitaries whose corresponding quasifree states are equivalent, the space $\Ker(J_0 + J_1)$ 
is finite-dimensional and has the structure of an ungraded Clifford module. Using the 
Atiyah--Bott--Shaprio isomorphism~\cite{ABS}, the corresponding 
Clifford module index gives a $K$-theoretic obstruction for the quasifree states to be equivalent when restricted 
to the even subalgebra of $A^\mathrm{car}_\mathrm{sd}(\calH,\Gamma)$. 
In Section \ref{sec:Basic_Quasifree_index} we review these ideas and their extensions to quasifree 
pure states that are symmetric with respect to a compact group~\cite{Matsui87, CareyEvans}.

Our first task  is to extend such Clifford module indices to an index with range $K$-theory of a $C^*$-algebra $A$
with real structure $\frakr$. The main technical tool we use to define these indices is the relative 
Cayley transform considered in~\cite{BKRCayley} for pairs of unitaries acting on a Hilbert 
$A$-module and whose difference is a compact endomorphism. This construction is reviewed and slightly 
extended in Section \ref{Sec:Fred_and_Cayley}.
We then apply this construction in Section \ref{sec:KO_index_general} to pairs of gapped BdG Hamiltonians 
$(H_0,H_1)$ acting on $(\calH, \Gamma)$ such that 
$\mathrm{sgn}(H_0) - \mathrm{sgn}(H_1) \in A$.

 In the case that $A=C^*(X)$, we show in Section \ref{sec:Coarse_and_KHom} 
that the condition $\mathrm{sgn}(H_0) - \mathrm{sgn}(H_1) \in C^*(X)$ is satisfied when 
$H_0$ and $H_1$ are quasilocal and the pure quasifree states of 
$A^\mathrm{car}_\mathrm{sd}(\calH,\Gamma)$ constructed from $H_0$ and $H_1$ are locally equivalent 
with respect to the Real representation $C_0(X) \to \calB(\calH)$. 
Because of the close connection between coarse $C^*$-algebras and duality theory, 
pairs of locally equivalent quasifree states can be used to construct both a 
$KO_\ast(C^*(X)^\frakr)$-index  and a 
$K$-homology class for $C_0(X)$. Using a description of the assembly map via duality theory 
and boundary maps in $K$-theory as developed by Roe~\cite{Roe02, Roe04}, 
our main result is that the coarse assembly 
map relates our constructed $K$-homology and $K$-theory elements. Compact and discrete 
group symmetries can also be incorporated with minor adjustments.

Because we work in the category of complex $C^*$-algebras with a real structure, the assembly map 
has a natural description using 
van Daele $K$-theory~\cite{vanDaele1, vanDaele2}, which we review in 
 in Section \ref{subsec:vanDaele_and_boundary}. In particular, 
building from~\cite[Section 5.2]{BKRCayley}, we write down an explicit representative of 
the boundary map in van Daele $K$-theory composed with the equivalence to $KKR$-theory. 
Once all the relevant objects are in place, our main result follows relatively easily from this 
general boundary map computation. The boundary map computation can also be applied to 
systems with a defect that is mathematically encoded by a semi-split short exact sequence (e.g. a 
codimension $1$ boundary). We lay the mathematical framework to study such systems in Section \ref{sec:DefectsAndSES}, 
though  leave a full treatment to another place.

Coarse geometry methods have already been effectively utilised to study free-fermionic 
topological phases~\cite{Kubota17,EwertMeyer, LT20a, LT20b}. 
It would also be interesting to consider analogous methods for more general quasifree dynamics and 
states such as those defined for Hilbert $C^*$-bimodules and their corresponding Toeplitz and Cuntz--Pimnser algebras~\cite{LN04}.

\subsection*{Applications to topological phases}

Gapped BdG Hamiltonians on Nambu space define quasifree ground states of the CAR algebra and 
provide an effective description  of free-fermionic systems.
Adopting a framework analogous to the study of symmetry protected topological (SPT) phases of unique 
gapped ground 
states, we consider a compact group $G$ and $G$-symmetric ground states which are equivalent but need 
not be $G$-equivariantly equivalent. When $G$ corresponds to physical (Altland--Zirnbauer) symmetries, 
the topological obstruction to connect these ground states is given by a Clifford module index. 
More generally, 
we can use results from Matsui and Carey--Evans \cite{Matsui87, CareyEvans} to give a $KO_2^G(\R)$-valued obstruction. 
We extend this work to construct $KO_2^G(A^\frakr)$-valued 
indices, which   provide a topological obstruction to connect pairs BdG Hamiltonians and ground states 
with respect to an auxiliary $C^{\ast,\frakr}$-algebra $A$ of allowed deformations.

We then consider a coarse space $X$ and 
pairs of locally equivalent 
gapped ground states with respect to 
a representation of $C_0(X)$ on the Nambu space.
The coarse index 
$\mu_X^G:KO^{-\ast}_G(C_0(X)^\frakr) \to KO_\ast^G( C^*(X)^\frakr )$ then 
gives a topological obstruction to connect locally equivalent $G$-symmetric
gapped ground states via a path of gapped ground states that respects the 
$G$-symmetry and is local with respect to the representation of $C_0(X)$. 
For the case of a discrete group $\Upsilon$ acting isometrically and cocompactly on $X$, the range of 
the assembly map is $KO_\ast( C^*_r(\Upsilon))$, which directly connects to more standard 
approaches to free-fermionic phases of matter via $K$-theory.

Our result  provides a new and potentially useful approach for studying local gapped 
free-fermionic phases. 
Coarse geometry methods have been used to consider  interacting gapped ground states 
by Kapustin, Sopenko and Spodyneiko~\cite{KapustinSopenko, KapustinThermal}, so our framework
may also be useful beyond the free-fermionic setting.

\subsection*{Outline}
We collect some basic facts on Fredholm operators and Kasparov theory in Section \ref{Sec:Fred_and_Cayley}. 
Because gapped quasifree ground states with physical (Altland--Zirnbauer) symmetries can be described via 
Real mutually anti-commuting skew-adjoint unitaries~\cite{KZ16}, we also extend some results on the Cayley transform of unitaries~\cite{BKRCayley} 
to this setting. 
 The Cayley transform provides a way to pass between $KK$-theory and van Daele $K$-theory, 
which we also introduce as well as its application to boundary maps in Kasparov theory.

Section \ref{sec:Basic_Quasifree_index} reviews pure quasifree states of the self-dual CAR algebra $A^{\mathrm{car}}_\mathrm{sd}(\calH,\Gamma)$ 
and the construction of Clifford module indices studied in~\cite{Matsui87,CareyEvans} that characterise 
pairs of symmetric quasifree states. This is extended in Section \ref{sec:KO_index_general} to $KO_\ast(A^\frakr)$-valued 
indices and we compute the image of such indices under the boundary map from a semi-split
short exact sequence.

Finally in Section \ref{sec:Coarse_and_KHom} we consider coarse spaces, pseudolocal gapped BdG Hamiltonians, 
locally equivalent quasifree states and their topological description via $K$-homology and $K$-theory. 
The coarse assembly map relates these pictures and we briefly consider compact symmetries and discrete cocompact symmetries.

\section{Preliminaries on Index theory and the Cayley transform} \label{Sec:Fred_and_Cayley}

\subsection{Kasparov modules and $KKR$-theory}

We will primarily  work in the category of Real $C^*$-algebras or $C^{\ast, \mathfrak{r}}$-algebras, 
which are  complex $C^*$-algebras with a real structure, an anti-linear 
order-$2$ automorphism $a\mapsto a^{\frakr_A}$ such that 
$(a^*)^{\frakr_A} = (a^{\frakr_A})^*$ for all $a \in A$. We say that 
$a\in A$ is Real if $a^{\frakr_A} = a$. If $A$ has a $\Z_2$-grading $A= A^0 \oplus A^1$ we also 
assume that $(A^i)^{\frakr_A} \subset A^i$, $i \in \{0,1\}$. We recover the complex theory by 
ignoring the real structure $\mathfrak{r}_A$. Similarly, restricting to the subalgebra $A^{\frakr_A} = \{a \in A\,:\, a^{\frakr_A} = a\}$ 
gives a real $C^*$-algebra,  a $C^*$-algebra over the number field $\R$.
When the context is unambiguous, we will  write $\frakr_A$ as $\frakr$.

\begin{example}[Real Clifford algebras]
Given $r, s \in \mathbb{N}$, the Real $\Z_2$-graded Clifford algebra $\Cl_{r,s}$ is the complex $C^*$-algebra generated by 
the elements 
$\{\gamma_1,\ldots, \gamma_r, \rho_1,\ldots ,\rho_s\}$,
which are odd, mutually anti-commute and
\[
  \gamma_j = \gamma_j^\frakr = \gamma_j^*, \qquad \gamma_j \gamma_k + \gamma_k \gamma_j = 2\delta_{j,k}, \qquad \rho_j = \rho_j^\frakr = -\rho_j^*, \qquad \rho_j \rho_k + \rho_k \rho_j = -2\delta_{j,k}.
\]
As complex algebras $\Cl_{r,s} \cong \Cl_{r+s}$. The real Clifford algebra $Cl_{r,s}$ is algebraic span of 
$\{\gamma_1,\ldots, \gamma_r, \rho_1,\ldots ,\rho_s\}$ over $\R$, where $\Cl_{r,s}^{\frakr} = Cl_{r,s}$. 

We will often make use of the isomorphism $\End(\bigwedge^* \C) \cong \Cl_{1,1}$ with Real generators $\gamma$ and $\rho$. 
More generally, $\End(\bigwedge^* \C^n) \cong \Cl_{n,n}$.

We will occasionally  consider ungraded Clifford algebras, though we reserve the notation $\gamma$ and $\rho$ for 
odd generating elements. In particular, any Clifford algebra appearing in a Kasparov module will always be interpreted 
as $\Z_2$-graded.
\end{example}

We now briefly review Real Kasparov theory or $KKR$-theory~\cite{Kasparov80}. 
Unless otherwise stated, $B$ is a $\sigma$-unital $C^{*,\frakr}$-algebra and 
$E_B$ is a countably generated right Hilbert $B$-module, see~\cite{Lance} for the basic theory. 
 We will call such $B$-modules  Hilbert $C^*$-modules.
We denote by $\End_B(E)$ and 
$\mathbb{K}_B(E)$ the adjointable and compact operators respectively.
In the special case  where $E=B$ as a vector space with right-action by 
right-multiplication and $(b_1 \mid b_2)_B = b_1^* b_2$, $b_1,b_2 \in B$, we have that 
$\mathbb{K}_B(B) = B$ and $\End_B(B) = \mathrm{Mult}(B)$, the multiplier algebra of $B$.

A complex Hilbert $C^*$-module 
$E_B$ is a Real Hilbert $C^*$-module if 
there is an antilinear map $\mathfrak{r}_E:E_B\to E_B$,
called the real involution, 
such that for all $e, e_1,e_2 \in E_B$ and $b \in B$, 
\[
  (e^{\rs_E})^{\rs_E} = e, \qquad \qquad e^{\mathfrak{r}_E} \cdot b^{\mathfrak{r}_B} = (e \cdot b)^{\mathfrak{r}_E}, 
  \qquad \qquad ( e_1^{\mathfrak{r}_E} \mid e_2^{\mathfrak{r}_E} )_B = \big(( e_1 \mid e_2 )_B\big)^{\mathfrak{r}_B}.
\] 
The real involution on $E_B$ induces 
a real structure $\mathfrak{r}$ on $\End_B(E)$ 
via $S^\mathfrak{r} e = \big( S(e^{\mathfrak{r}_E} ) \big)^{\mathfrak{r}_E}$ for any $e \in E_B$. 
Given a separable Real $C^*$-algebra $A$, any representation
$\pi: A \to \End_B(E)$ should be compatible with this real 
structure, $\pi(a^{\mathfrak{r}_A}) = \pi(a)^\mathfrak{r}$ for all $a \in A$.

 We will often work with unbounded operators on Hilbert $C^*$-modules, 
 see~\cite[Chapter 9]{Lance} for more details.  
 We  recall 
 that a densely defined closed 
right $B$-linear operator $D:\Dom (D)\subset E_B \to E_{B}$ 
is \emph{regular} if $D^*$ is densely defined and the operator $1+D^*D:\Dom (D^*D) \to E_{B}$ has dense range.  
 Note also that $\Dom(D)$ must be invariant under the right $B$-action in order to 
obtain a right $B$-linear operator $D: \Dom(D) \to E_B$.
We call $D$ Real and write $D^\rs = D$ if $(\Dom (D))^{\rs_E} \subset \Dom (D)$ and 
$(De^{\rs_E})^{\rs_E} = De$ for all $e \in \Dom (D)$.
We also recall the graded commutator, where for endomorphisms
$S,\,T$ with homogenous parity 
$[S,T]_\pm = ST - (-1)^{\mathrm{deg}(S)\,\mathrm{deg}(T)}TS$.

\begin{defn}
Let $A$ and $B$ be $\Z_2$-graded Real $C^{*}$-algebras.
A Real Kasparov  module $(A, {}_\pi E_B, F)$ consists of 
\begin{enumerate}
\item[(i)] A Real and
$\Z_2$-graded Hilbert $C^*$-module ${E}_B$, 
\item[(ii)] A Real and $\Z_2$-graded $*$-homomorphism $\pi:A \to \End_B(E)$, 
\item[(iii)] A self-adjoint and odd operator $F=F^\rs \in \End_B(E)$  such that $[F,\pi(a)]_\pm, \, \pi(a)(\one - F^2) \in \mathbb{K}_B(E)$ 
for all $a\in A$.
\end{enumerate}
If $0=[F,\pi(a)]_\pm = \pi(a)(\one - F^2)$ for all $a \in A$, we say that $(A, {}_\pi E_B, F)$ is degenerate.

An  unbounded Real Kasparov module is  a triple $(\calA, {}_\pi{E}_B, D)$ 
with $\calA \subset A$ a dense $\ast$-subalgebra such that conditions (i) and (ii) of a Real Kasparov 
module are satisfied and  (iii) is replaced by the condition:
\begin{enumerate}
\item[(iii')] There is an unbounded self-adjoint, regular and odd operator $D=D^\rs$  such that for all $a\in \calA$, 
$\pi(a)\Dom(D)\subset\Dom(D)$ and
\begin{align*}
  & [D,\pi(a)]_\pm \in  \End_B(E),   &&\pi(a)(1+D^2)^{-1/2} \in  \mathbb{K}_B(E).
  \label{eq:defn}
\end{align*}
\end{enumerate}
\end{defn}

When $B=\C$, bounded and unbounded Kasparov modules are also called Fredholm modules and spectral triples respectively.

We will often omit the representation $\pi:A\to \End_B(X)$ if the context is clear. 
If $(\calA, E_B, D)$ is an unbounded Real Kasparov module, then the results of~\cite{BJ} can be easily 
adapted to the Real setting to show that $(A, E_B, F_D)$ is a Real
Kasparov module for $F_D=D(1+D^2)^{-1/2}$.
Equivalence classes of Real Kasparov modules give an abelian group $KKR(A,B)$~\cite{Kasparov80},
though this group depends on the choice of real structures for $A$ and $B$. 
Degenerate Kasparov modules represent the group identity of $KKR(A,B)$.

If $(A, E_B, F)$ is a Real  Kasparov module, then we can 
ignore the real structures and obtain a complex  Kasparov module and 
class in $KK(A, B)$. 
If we restrict the Real Hilbert $C^*$-module $E_B$ to the 
elements fixed under $\rs_E$, we obtain a real Hilbert $C^*$-module $E_{B^{\rs_B}}^{\rs_E}$. 
Similarly, the Real left action of $A$ becomes a real left action 
$\pi: A^{\rs_A} \to \End_{B^{\rs_B}}( E^{\rs_E} )$. 
We do not lose any information by restricting Real Kasparov modules to real Hilbert $C^*$-modules 
and algebras. Similarly, 
real Kasparov modules can be complexified to obtain Real Kasparov modules and 
$KKR(A, B) \cong KKO(A^{\rs_A}, B^{\rs_B})$.

The Clifford algebras $\Cl_{r,s}$ also play an important role in the $KKR$-groups, where we have that 
$KKR(A \hox \Cl_{r,s}, B ) \cong KKR( A, B \hox \Cl_{s,r})$. This isomorphism is obtained by 
the following composition
\[
  KKR(A \hox \Cl_{r,s}, B ) \xrightarrow{ \hox \mathrm{Id}_{\Cl_{s,r}}} KKR( A \hox \Cl_{r,s} \hox \Cl_{s,r}, B \hox \Cl_{s,r} ) 
  \to KKR( A, B \hox \Cl_{s,r}),
\]
where
the first map is the external product with the Kasparov module $\big(\Cl_{s,r},  {\Cl_{s,r}}_{\Cl_{s,r}}, 0 \big)$ and is 
an isomorphism of $KKR$-groups. The second map comes from the identification 
$\Cl_{r,s} \hox \Cl_{s,r} \cong \End\big( \bigwedge^* \C^{r+s} \big) \cong M_{2^{r+s}}(\C)$ and the stability of $KKR$.

If the algebra $B$ is trivially graded, $B^1=\{0\}$, we can also 
consider real $K$-theory, where $KKR(\Cl_{r,s}, B) \cong KKO(Cl_{r,s}, B^{\rs}) \cong KO_{r-s}(B^{\rs})$. 
Similarly, the real $K$-homology groups of a trivially graded algebra $A$ can be expressed 
as $KKR(A \otimes \Cl_{r,s}, \C) \cong KO^{s-r}(A^\frakr)$.

Finally we consider the case of group actions and equivariant Kasparov modules. 
Fix a compact or discrete group $G$ and an action  $\beta: G \to \Aut(B)$. 
We say that $\beta$ is Real and $\Z_2$-graded 
if $\beta_g(b^\rs) = \beta_g(b)^\rs$ and $\beta_g( B^j) \subset B^j$ for all 
$g \in G$, $b\in B$ and $j \in \{0,1\}$. 
A Real Hilbert $C^*$-module 
$E_B$ is $G$-equivariant if there is a
homomorphsim $\eta$ from $G$ into the invertible and bounded (not necessarily adjointable) linear transformations 
on $E$  that preserves the $\Z_2$-grading 
and is such that 
\[
  \eta_g(e^\frakr) = \eta_g(e)^\frakr, \qquad 
   \eta_g( e \cdot b) = \eta_g(e) \cdot \beta_g(b), \qquad   ( \eta_g(e_1) \mid \eta_g(e_2) )_B = \beta_g\big( ( e_1 \mid e_2)_B \big) 
\]
for all $e,e_1,e_2 \in E$, $b \in B$ and $g\in G$.
Such an action then induces a Real and $\Z_2$-graded action $\tilde{\eta}: G\to \Aut\big( \End_B(E) \big)$ where 
$\tilde{\eta}(T)e = \eta_g \circ T \circ \eta_{g^{-1}}(e)$ for any $T \in \End_B(E)$, $e \in E_B$ and $g \in G$. 
If $A$ is a Real $C^*$-algebra with a Real and $\Z_2$-graded group action $\alpha$, we require that 
any representation $\pi:A \to \End_B(E)$  be equivariant with respect to $\alpha$ and $\tilde{\eta}$.
We say that $T:\Dom(T)\subset E_B\to E_B$ 
is $G$-invariant if $\eta_g ( \Dom(T) ) \subset \Dom(T)$ and $\tilde{\eta}_g(T) = T$ for all $g \in G$. 

With these preliminaries in place, a $G$-equivariant (unbounded) Real Kasparov module is a Real (unbounded) 
Kasparov module with an equivariant Hilbert $C^*$-module and left-action such that  the 
self-adjoint operators $F$ or $D$ are  $G$-invariant.

\subsection{Fredholm operators on Hilbert $C^*$-modules}
We briefly provide some further information on Fredholm theory in the Hilbert $C^*$-module setting. 
A more comprehensive treatment can be found in~\cite{Joachim, vdDungen17}. 
We fix a $\sigma$-unital $C^{\ast, \frakr}$-algebra $B$ and a Real countably generated 
 Hilbert $C^*$-module $E_B$.

\begin{defn}
Let $S$ be a regular operator on $E_B$. We say that $S$ is Fredholm if there is a 
parametrix $Q \in \End_B(E)$ such that $SQ$ and $QS$ are closable with adjointable 
closures and $\ol{QS} -1$, $\ol{SQ}-1 \in \mathbb{K}_B(E)$.
\end{defn}

If $T \in \End_B(E)$ (so $T$ is bounded), then   $T$ is Fredholm if and only if  $q(T) \in \calQ_B(E)$ is invertible with 
$\calQ_B(E) = \End_B(E)/\mathbb{K}_B(E)$ the Calkin algebra of the Hilbert $C^*$-module $E_B$ and 
$q: \End_B(T) \to \calQ_B(E)$ the quotient map.

\begin{prop}[{\cite[Lemma 2.2]{Joachim}, \cite[Proposition 2.14]{vdDungen17}}] \label{prop:Fred_to_KasMod}
If $S=S^\frakr$ is a skew-adjoint Fredholm operator on a trivially graded Real Hilbert $C^*$-module $E_B$, 
then the triple 
\[
   \Big( \Cl_{1,0}, \, E_B \otimes \bigwedge\nolimits^{\! *} \C, \, S(\one - S^2)^{-1/2} \otimes \rho \Big)
\]
is a Real Kasparov module, where the left $\Cl_{1,0}$-action is generated by $\one \otimes \gamma$. 
\end{prop}

In the case that $S$ has a compact resolvent, the the class $[S] \in KKR(\Cl_{1,0}, B)$ from 
Proposition \ref{prop:Fred_to_KasMod} can be directly represented by the unbounded Real Kasparov module,
$$
  [S] = \Big[ \big( \Cl_{1,0}, \, E_B  \otimes \bigwedge\nolimits^{\! *} \C, \, S  \otimes \rho \big) \Big] \in KKR(\Cl_{1,0},B) \cong KO_1(B^\frakr).
$$

\subsection{The Cayley transform of odd self-adjoint unitaries} \label{subsec:CayleyOdd}

Let $A$ be a $\Z_2$-graded and $\sigma$-unital $C^{*,\frakr}$-algebra and 
$E_A$ a countably generated and $\Z_2$-graded Real Hilbert $C^*$-module over $A$. 
We suppose that $\End_A(E)$ contains as many odd self-adjoint unitaries as we need. 
We can always ensure this by taking a graded tensor product $E'_A = E_A \hox \bigwedge^* \C^n$, where 
$\End_A( E' ) \simeq \End_A(E) \hox \Cl_{n,n}$.
Let us then fix a representation of $\Cl_{k,0}$ on $E_A$ with generators $\{\gamma_j\}_{j=1}^k$.
We are interested in the set
\begin{equation} \label{eq:OSU_endo_group}
 \calO_{E_A}^k = \big\{ V \in \End_A(E) \,:\, V\text{ odd,}\,\, V=V^*=V^{-1} = V^\mathfrak{r}, \,\, V\gamma_j = -\gamma_j V \,\text{ for } j=1,\ldots,k \}.
\end{equation}

\begin{lemma}[{cf. \cite[Lemma 4.5]{BKRCayley}}] \label{lem:Cayley_selfadj}
Given $V_0,V_1 \in \calO_{E_A}^k$ with $\|V_0 -V_1\|_{\calQ_A(E)} < 2$, define the unbounded operator 
$$
 \mathcal{C}_{V_0}(V_1) = V_0(V_1+V_0)(V_1-V_0)^{-1}, \qquad \qquad
 \Dom(\mathcal{C}_{V_0}(V_1)) = (V_1-V_0)E_A.
$$
Then $\mathcal{C}_{V_0}(V_1)$ is odd, self-adjoint, Real, regular and 
anti-commutes with $\{V_0,\gamma_1,\ldots,\gamma_{k}\}$ on 
$\ol{(V_0-V_1)E}_A$, the closure of $\Dom(\mathcal{C}_{V_0}(V_1))$ in the 
module norm of $E_A$. Furthermore, 
 $F_{\calC_{V_0}(V_1)} = \calC_{V_0}(V_1)(1+ \calC_{V_0}(V_1)^2)^{-1/2}$ satisfies
$\| \one - F_{\calC_{V_0}(V_1)}^2 \|_{\calQ_A(E)} < 1$.
\end{lemma}
\begin{proof}
 We first note that because $V_0$ and $V_1$ are self-adjoint unitaries, 
$V_0(V_1 \pm V_0) = (V_0 \pm V_1)V_1$. In particular, for any $e \in E_A$,
\[
  V_0 ( V_1 - V_0) e =  -(V_1-V_0)V_1 e \in (V_1-V_0)E_A.
\]
and so $V_0$ preserves the domain of $\calC_{V_0}(V_1)$. 
Because $V_0$ and $V_1$ anti-commute with $\{\gamma_1,\ldots,\gamma_{k}\}$, 
we see that $\{\gamma_j\}_{j=1}^k$ preserve $\Dom(\mathcal{C}_{V_0}(V_1))$ 
and a simple computation gives that these operators anti-commute with 
$\calC_{V_0}(V_1)$.  We similarly have that on $\Dom(\calC_{V_0}(V_1))$
\begin{align*}
  V_0 \, \calC_{V_0}(V_1) &= (V_1 + V_0) V_1 V_1 (V_1 - V_0)^{-1} 
  = V_0(V_1 + V_0) \big( (V_1-V_0)V_1 \big)^{-1} \\
  &= V_0(V_1 + V_0) \big( V_0(V_0 - V_1) \big)^{-1} 
  = - V_0 (V_1+ V_0)( V_1-V_0) V_0 \\
  &= - \calC_{V_0}(V_1) \,V_0.
\end{align*}
It is immediate that $\calC_{V_0}(V_1)$ is Real and odd. 

To prove the that $\calC_{V_0}(V_1)$ is self-adjoint and regular, one considers the bounded 
transform $F_{\calC_{V_0}(V_1)} = \calC_{V_0}(V_1)(1+ \calC_{V_0}(V_1)^2)^{-1/2}$. 
To make sense of this operator, we first compute using the normality of $V_0V_1$
\begin{align*}
  \big( 1+ \calC_{V_0}(V_1)^2 \big)^{-1/2} &= \big( 1 + V_0(V_1V_0 +1)(V_1V_0 - 1)^{-1} V_0 (V_0V_1+1)(1-V_0V_1)^{-1} \big)^{-1/2} \\
  &=\big( 1 - (2 + V_1V_0 + V_0V_1)(-2 +V_1V_0+V_0V_1)^{-1} \big)^{-1/2} \\
  &= \big( ( 2 - V_1 V_0 - V_0 V_1 + 2 + V_1 V_0 + V_0 V_1) (2 - V_0V_1 - V_1V_0)^{-1} \big)^{-1/2} \\
  &= \big( 4 (2 - V_0 V_1 - V_1 V_0)^{-1} \big)^{-1/2} 
  = \frac{1}{2} (2 - V_0 V_1 - V_1 V_0)^{1/2}.
\end{align*}
Therefore we can write $F_{\calC_{V_0}(V_1)} = \frac{1}{2} V_0(V_1V_0+1)(V_1V_0-1)^{-1}(2 - V_0 V_1 - V_1 V_0)^{1/2}$.
It is shown in~\cite[Lemma 4.5]{BKRCayley} that $F_{\calC_{V_0}(V_1)}$ is self-adjoint and 
has norm bounded by $1$. Then using that $V_0$ commutes with $V_0 V_1 + V_1 V_0$ and 
the normality of $V_1V_0$,
\begin{align}  \label{eq:Cayley_F^2}
  F_{\calC_{V_0}(V_1)}^2 &= \frac{1}{4} \big(  V_0(V_1V_0+1)(V_1V_0-1)^{-1}(2 - V_0 V_1 - V_1 V_0)^{1/2} \big)^2  \nonumber \\
  &= -\frac{1}{4} (V_1V_0+1)(V_0V_1+1)( V_1V_0-1)^{-1}(1-V_1V_0)^{-1} (2 - V_0 V_1 - V_1 V_0) \nonumber \\
  &= \frac{1}{4} (V_1V_0+1)(V_0V_1+1) = \frac{1}{4} ( 2+ V_0V_1 + V_1V_0).
\end{align}
We therefore have that
\begin{align} \label{eq:Cayley_1-F^2}
  \one - F_{\calC_{V_0}(V_1)}^2 &=  \frac{1}{2} - \frac{1}{4} \big( V_0V_1 + V_1V_0\big) 
  = \frac{1}{4} \big( 2 - V_0 V_1 - V_1 V_0 \big) 
  = \frac{1}{4}(V_0 -V_1)^2.
\end{align}
In particular, $\one - F_{\calC_{V_0}(V_1)}$ is positive and $(\one - F_{\calC_{V_0}(V_1)})^{1/2}$ has dense range 
$(V_1-V_0)E_A$. Applying~\cite[Theorem 10.4]{Lance}, 
$\calC_{V_0}(V_1)$ is self-adjoint and regular.

Finally using Equation \eqref{eq:Cayley_1-F^2},
\[
  \big\| \one - F_{\calC_{V_0}(V_1)}^2 \big\|_{\calQ_A(E)} = \frac{1}{4} \big\| (V_0-V_1)^2  \big\|_{\calQ_A(E)} 
  < 1.   \qedhere
\]
\end{proof}

The operator $\calC_{V_0}(V_1)$ maps $(V_1-V_0)E_A$ to $(V_1+V_0)E_A$ in analogy 
to the standard Cayley transform for unitary operators on Hilbert spaces. 
Because the operator $V_0-V_1$ need not be dense in $E_A$ nor have closed range,
the operator $\mathcal{C}_{V_0}(V_1)$ is a densely-defined unbounded operator on the submodule 
$\ol{(V_0-V_1)E}_A \subset E_A$.  One may consider $\calC_{V_0}(V_1)$ as densely defined 
right $A$-linear map $\ol{(V_0-V_1)E}_A \to E_A$ such that it is self-adjoint and regular on 
$\ol{(V_0-V_1)E}_A$. 
Because the operators $\{\gamma_j\}_{j=1}^k$ anti-commute with $V_0$ and $V_1$, they 
 restrict to mutually anti-commuting odd self-adjoint unitaries acting on $\ol{(V_0-V_1)E}_A$.

\begin{prop} \label{prop:OSU_pair_class}
Let $V_0,V_1 \in \calO_{E_A}^k$ with $\|V_0 -V_1\|_{\calQ_A(E)} < 2$. 
Then the triple
\[
  \Big( \Cl_{k+1, 0}, \, E_A, \, F_{\calC_{V_0}(V_1)} = \calC_{V_0}(V_1)(1+ \calC_{V_0}(V_1)^2)^{-1/2} \Big)
\]
is a Real Kasparov module with left Clifford generators $\{V_0, \gamma_1,\ldots,\gamma_k\}$. 
If $V_0-V_1 \in \mathbb{K}_A(E)$, then the class in $KKR(\Cl_{k+1,0}, A)$ of this Kasparov 
module can be represented by the unbounded Kasparov module 
\[
   \Big( \Cl_{k+1,0}, \,  \ol{(V_1-V_0)E}_A, \, \mathcal{C}_{V_0}(V_1) \Big)
\]
with Clifford generators $\{V_0, \gamma_1,\ldots,\gamma_k\}$.
\end{prop}
\begin{proof}
By Lemma \ref{lem:Cayley_selfadj}, the estimate 
$\|V_0 -V_1\|_{\calQ_A(E)} < 2$ implies 
that $\| \one - F_{\calC_{V_0}(V_1)}^2 \|_{\calQ_A(E)} < 1$ and so 
$F_{\calC_{V_0}(V_1)}$ is invertible in the Calkin algebra and, hence, 
Fredholm. Because $\calC_{V_0}(V_1)$ anti-commutes with $\{V_0, \gamma_1,\ldots, \gamma_k\}$, 
so does $F_{\calC_{V_0}(V_1)}$. Thus the triple 
$\big( \Cl_{k+1, 0}, \, E_A, \, F_{\calC_{V_0}(V_1)}\big)$ is a Real Kasparov module.

Similar to Equations \eqref{eq:Cayley_F^2} and \eqref{eq:Cayley_1-F^2}, we compute that
\begin{align*}
  \one + \calC_{V_0}(V_1)^2 &= 1 + (2+ V_0V_1 + V_1V_0)( 2 - V_0V_1 - V_1V_0)^{-1} \\
    &= 4(2 - V_0V_1 - V_1V_0)^{-1} = 4(V_1-V_0)^{-2}.
\end{align*}
Therefore 
$(\one+\calC_{V_0}(V_1)^2)^{-1/2} = \frac{1}{2}|V_0-V_1|$, which will be compact 
if $V_0-V_1 \in \mathbb{K}_A(E)$. This result combined with Lemma \ref{lem:Cayley_selfadj} 
shows that $\big( \Cl_{k+1,0},  \, \ol{(V_1-V_0)E}_A,  \, \mathcal{C}_{V_0}(V_1) \big)$ is an 
unbounded Kasparov module and is an unbounded lift of 
$\big( \Cl_{k+1, 0}, \, E_A, \, F_{\calC_{V_0}(V_1)}\big)$.
\end{proof}

\subsection{The Cayley transform of skew-adjoint ungraded unitaries} \label{subsec:skew_cayley}

Fix a $\sigma$-unital, ungraded and Real $C^*$-algebra $B$ and 
an ungraded and countably generated Real Hilbert $C^*$-module $E_B$.  
We also suppose that there exist operators $\{\kappa_j\}_{j=1}^{n} \subset \End_B(E)$ 
such that for all $j,k \in \{1,\ldots,n\}$,
$$
  \kappa_j^*=-\kappa_j, \qquad \qquad  \kappa_j^\frakr = \kappa_j, \qquad \qquad \kappa_j \kappa_k + \kappa_k\kappa_j = -2\delta_{j,k}.
$$
Such an assumption can always be satisfied by taking an ungraded (Real) representation of 
$\Cl_{0,n}$ on $\C^{\nu}$ and considering $E'_B = E_B \otimes \C^{\nu}$.
We then define the set
$$
  \calU_{E_B}^n = \big\{ J \in \End_B(E)\,:\, J=J^\mathfrak{r} = -J^*,\,\,   J^2=-1, \,\, \kappa_j J = -J \kappa_j \,\text{ for all } j=1,\ldots,n \big\}.
$$

\begin{lemma}
Let $J_0,J_1 \in \calU_{E_B}^n$ be such that $\|J_0 - J_1 \|_{\calQ_B(E)} < 2$. Define 
the operator
$$
  \calC_{J_0}(J_1) = J_0(J_1 +J_0)(J_1 - J_0)^{-1}, \qquad 
  \Dom\big(\calC_{J_0}(J_1)\big) = (J_1 - J_0)E_B.
$$
Then $\calC_{J_0}(J_1)$ is an unbounded, Real, regular and skew-adjoint operator on 
$\ol{(J_1-J_0)E}_B$ that 
anti-commutes with $\{J_0,\kappa_1,\ldots,\kappa_{n}\}$. 
\end{lemma}
\begin{proof}
Given $\{J_0,J_1,\kappa_1,\ldots,\kappa_{n}\}$ acting on $E_B$ we can consider 
$\{J_0\otimes \rho, J_1\otimes \rho, \kappa_1\otimes \rho,\ldots,\kappa_{n}\otimes \rho\}$ acting on 
$(E\otimes \Cl_{0,1})_{B\otimes \Cl_{0,1}}$ and with $\rho$ the skew-adjoint generator. 
All operators are now odd self-adjoint unitaries and so we can apply Lemma \ref{lem:Cayley_selfadj}. 
Expressing these results in terms of operators 
on $E_B$, we get the desired results, e.g.,
$\calC_{J_0\otimes \rho}(J_1\otimes \rho) = \calC_{J_0}(J_1) \otimes \rho$, so the 
self-adjointness and regularity of $\calC_{J_0\otimes \rho}(J_1\otimes \rho)$ gives the 
skew-adjointess and regularity of $\calC_{J_0}(J_1)$.
\end{proof}

Note that because $\kappa_j$ anti-commute with $J_0$ and $J_1$ for all $j\in \{1,\ldots,n\}$, the operators 
$\kappa_j$ also restrict to skew-adjoint unitaries on the submodule $\ol{(J_0-J_1)E}_B$. 
An adaptation of Proposition \ref{prop:OSU_pair_class} to the ungraded and skew-adjoint setting gives the following.

\begin{prop} \label{prop:skew_cayley_kk}
Let $J_0,J_1 \in \calU_{E_B}^n$ be such that $\|J_0 - J_1 \|_{\calQ_B(E)} < 2$. 
Then the triple
\[
  \Big( \Cl_{n+2,0}, \, E_B \otimes \bigwedge\nolimits^{\! *} \C, \, 
     \calC_{J_0}(J_1)(\one  - \calC_{J_0}(J_1)^2)^{-1/2}  \otimes \rho \Big)
\]
is a Real Kasparov module, where the $\Cl_{n+2,0}$-action has generators 
$\{\one \otimes \gamma, J_0 \otimes \rho, \kappa_1\otimes \rho, \ldots, \kappa_n \otimes \rho\}$. 
If $J_0 - J_1 \in \mathbb{K}_B(E)$, then the corresponding class in $KKR(\Cl_{n+2}, B)$ can be 
represented by the unbounded Real 
Kasparov module
$$
  \Big( \Cl_{n+2,0}, \, \ol{(J_1-J_0)E}_B \otimes \bigwedge\nolimits^{\! *} \C, \, \calC_{J_0}(J_1)\otimes \rho \Big),
$$
with $\Cl_{n+2,0}$-generators $\{\one \otimes \gamma, J_0 \otimes \rho, \kappa_1\otimes \rho, \ldots, \kappa_n \otimes \rho\}$.
\end{prop}

We also list a few properties of the bounded transform $F_{\calC_{J_0}(J_1)} = \calC_{J_0}(J_1)(\one  - \calC_{J_0}(J_1)^2)^{-1/2}$. 
Completely  
analogous computations to those in Lemma \ref{lem:Cayley_selfadj}  give that 
$( \one - \calC_{J_0}(J_1)^2 \big)^{-1/2} = \frac{1}{2} ( 2 + J_1J_0 + J_0 J_1)^{1/2}$ and so 
\begin{equation}  \label{eq:Skew_Cayley_F}
       F_{\calC_{J_0}(J_1)} = \frac{1}{2} J_0 (J_1J_0 -1) (J_1 J_0 - 1)^{-1} ( 2 + J_1J_0 + J_0 J_1)^{1/2},
\end{equation}
which similar to  Equations \eqref{eq:Cayley_F^2} and \eqref{eq:Cayley_1-F^2} has the  properties 
\begin{equation} \label{eq:Skew_Cayley_F_properties}
  F_{\calC_{J_0}(J_1)}^2 = \frac{1}{4}( -2 + J_1J_0 + J_0 J_1), \qquad \qquad 
  \one + F_{\calC_{J_0}(J_1)}^2 = - \frac{1}{4} (J_0 - J_1)^2.
\end{equation}

\subsection{Van Daele $K$-theory and boundary maps} \label{subsec:vanDaele_and_boundary}

As it will be useful for several of our results below, we give a brief overview of van Daele $K$-theory, 
first considered in~\cite{vanDaele1, vanDaele2} and 
then further developed in~\cite{Roe04, Kellendonk15, Kubota16, BKRCayley}.

\begin{defn}
Let $A$ be a complex $C^*$-algebra. We say that $A$ has a balanced $\Z_2$-grading 
if $A$ contains an odd self-adjoint unitary. That is, there is an odd element $e$ satisfying $e=e^* = e^{-1}$. 
In particular, $A$ is unital. If $A$ has a real structure $\rs_A$, we also require $e^{\rs_A} = e$.
\end{defn}

For simplicity, we will assume that any $\Z_2$-graded and unital $C^*$-algebra $A$ is balanced 
graded, taking the tensor product $A = A' \hox \Cl_{1,1}$ if necessary. 
We can extend the grading and real structure of $A$ to $M_k(A)$ entrywise.

Let $V(A)=\bigsqcup_k\pi_0\big(\mbox{\rm OSU}(M_k(A))\big)$ be the disjoint union of homotopy classes of odd self-adjoint unitaries in $M_k(A)$, 
which is an abelian semigroup under direct summation,
$[x]+[y] = [x\oplus y]$. The Grothendieck group obtained from this semigroup will be denoted 
$GV(A)$. 
The semigroup homomorphism $d:V(A)\to \N$ taking the 
value $k$ on $M_k(A)$ induces a group
homomorphism $d:GV(A)\to\Z$.

\begin{defn} 
\label{defn:vD-K}
If $A$ is unital and has a balanced $\Z_2$-grading, then
the van Daele $K$-theory group of $A$ is $DK(A)=\Ker(d:GV(A)\to\Z)$.

If $A$ is not unital then we set
$DK(A)=\Ker(q_*:DK(A^\sim)\to DK(\C))$ where $q:A^\sim\to\C$ quotients the minimal unitisation $A^\sim$ by the ideal $A$.

Elements of $DK(A)$ are formal differences of odd self-adjoint unitaries, denoted by $[x]-[y]$. 
\end{defn}

It will also be useful to consider van Daele $K$-theory relative to a choice of basepoint~\cite{vanDaele1}.
For a balanced graded algebra $A$ and $e \in A$ an odd self-adjoint unitary, 
we let $V_e(A) =  \bigcup_{k}\pi_0\big(\mbox{\rm OSU}(M_{k}(A))\big)$ where we embed 
$M_k(A)$ into $M_{k+1}(A)$ via $x\mapsto x\oplus e$. 
Van Daele's $K$-theory group with a basepoint is defined as the Grothendieck
group $DK_e(A) = GV_e(A)$

The group $DK_e(A)$ does not depend on the choice of $e$ up to isomorphism~\cite[Proposition 2.12]{vanDaele1}.  
For any choice of basepoint $e$, $DK_e(A)\cong DK(A)$~\cite[Section 2.1.1]{BKRCayley}.

\begin{examples}
\begin{enumerate}
  \item If $A$ is a unital and trivially graded 
algebra, then odd self-adjoint unitaries of $A \otimes \Cl_{1,1}$ are of the form
$$
U = \frac{1}{2} (u+u^*) \otimes \gamma + \frac{1}{2} (u-u^*)\otimes \rho = \begin{pmatrix} 0 & u^* \\ u & 0 \end{pmatrix}
$$ 
where $u \in A$ is unitary and we make have made the identification $\{\gamma,\rho\} \sim \{\sigma_1, -i\sigma_2\}$. 
If $U^\frakr = U$, then $u^\frakr = u$ and 
the map $U\mapsto u$ 
furnishes an isomorphism $DK_{1\otimes \gamma}(A \otimes \Cl_{1,1}) \cong KO_1(A^\frakr)$.

  \item If $A$ is trivially graded and unital, then odd self-adjoint unitaries in $A \otimes \Cl_{1,0}$ take the 
  form $U = x \otimes \gamma$. Hence $x=x^*=x^\frakr$ is an even unitary in $A$ and the map 
  $DK_{1\otimes \gamma}(A) \ni [x \otimes \gamma] \mapsto \big[\tfrac{1-x}{2} \big] \in KO_0(A^\rs)$ is an isomorphism.
\end{enumerate}
\end{examples}

For  a balanced graded algebra $A$ with a closed, two-sided and graded ideal $I$ we 
define the relative van Daele group
$$
DK(A,\, A/I):=\{[x]-[y]:\,\, x,y\in \mbox{\rm OSU}(M_n(A)),\ \
x-y\in M_n(I)\}.
$$
Here $[\cdot]$ denotes homotopy classes in $\mbox{\rm OSU}(M_n(A))$.
As expected, there is an excision isomorphism $DK(I) \cong DK(A, A/I)$~\cite[Proposition 2.4]{BKRCayley}.

The excision isomorphism gives us a basepointed description of the van Daele $K$-theory for non-unital 
$C^*$-algebras $A$ such that $\mathrm{Mult}(A)$ is balanced graded. For such algebras we fix an odd 
self-adjoint unitary $e \in \mathrm{Mult}(A)$ and let 
$A^{\sim e}$ be the subalgebra of $\mathrm{Mult}(A)$ generated by $A$ and $e$. We  may then consider 
$DK_e(A) = DK( A^{\sim e}, \, A^{\sim e}/A)$, see~\cite[Section 2.1]{BKRCayley} for the full details.

\begin{remark} \label{rk:extra_clifford_gens_shifts_the_degree}
Let $A$ be balanced graded and fix the $\Cl_{k,0}$-generators $\{\gamma_1,\ldots, \gamma_k\}\subset A$. 
Recalling $\calO_{E_A}^k$ from Equation \eqref{eq:OSU_endo_group} on Page \pageref{eq:OSU_endo_group}, 
we may also wish to consider 
homotopy classes of odd self-adjoint unitaries in the set
\[
  \calO_A^k = \big\{ V = V^* = V^\frakr = V^{-1} \in A \,:\, V \text{ odd, }  V\gamma_j = -\gamma_j V \,\text{ for } j=1,\ldots,k  \big\}.
\]
Extending to matrices, we can define another semigroup of homotopy classes of odd self-adjoint unitaries 
in $\bigoplus_n M_n(A)$ that anti-commute with $\{\gamma_j^{\oplus n}\}_{j=1}^k$.
However, noting that $\Cl_{k,0} \hox \Cl_{0,k} \cong \End(\bigwedge^* \C^k) \cong M_{2^k}(\C)$ and that the 
representations of $\Cl_{k,0}$ and $\Cl_{0,k}$ on $\bigwedge^*\C^k$ graded-commute, homotopy classes of odd self-adjoint unitaries in 
$\bigoplus_n M_n(A)$ that anti-commute (graded-commute) with $\{\gamma_j^{\oplus n}\}_{j=1}^k$ are equivalent 
to homotopy classes of odd self-adjoint unitaries in $\bigoplus_m M_m( A \hox \Cl_{0,k})$. 
Hence from the perspective of $K$-theory, it suffices to consider  $DK(A \hox \Cl_{0,k})$.
\end{remark}

We can use the Cayley transform of odd self-adjoint unitaries from Section \ref{subsec:CayleyOdd} to relate 
van Daele $K$-theory to $KKR$-theory.

\begin{thm}[{\cite[Theorem 4.15]{BKRCayley}}] \label{thm:Cayley_iso_extra_Clifford}
Let $A$ be a Real $C^*$-algebra such that $\mathrm{Mult}(A)$ is balanced graded. Suppose that 
$V_0, \, V_1 \in \mathrm{Mult}(A)$ are odd self-adjoint unitaries anti-commuting with 
the $\Cl_{k,0}$-generators  $\{\gamma_1,\ldots, \gamma_k\}\subset \mathrm{Mult}(A)$ and such that 
$V_0 - V_1 \in A$. Then the unbounded Real Kasparov module from Proposition \ref{prop:OSU_pair_class},
\[
   (V_0,V_1) \mapsto   \big( \Cl_{k+1,0}, \, \ol{(V_1-V_0)A}_A, \, \mathcal{C}_{V_0}(V_1) \big) 
\]
gives an isomorphism
$DK(\mathrm{Mult}(A) \hox \Cl_{0,k}, \, \mathrm{Mult}(A) \hox \Cl_{0,k} / A\hox \Cl_{0,k}) \cong KKR(\Cl_{k+1,0}, A)$.
\end{thm}
\begin{proof}
Recalling Remark \ref{rk:extra_clifford_gens_shifts_the_degree}, because $V_0$ and $V_1$ anti-commute 
with $\{\gamma_j\}_{j=1}^k$, they give a class in the degree shifted $DK( \mathrm{Mult}(A) \hox \Cl_{0,k})$ and 
because $V_0 - V_1 \in A$ we can take the relative class 
$[V_0]-[V_1] \in DK(\mathrm{Mult}(A) \hox \Cl_{0,k}, \, \mathrm{Mult}(A) \hox \Cl_{0,k} / A\hox \Cl_{0,k})$. 
Recalling the excision isomorphism, the cited result in~\cite{BKRCayley} then finishes the proof.
\end{proof}

Theorem \ref{thm:Cayley_iso_extra_Clifford} also shows that if $A$ is trivially graded, then 
$DK(A\otimes \Cl_{r,s}) \cong KO_{1+s-r}(A^\frakr)$, see also~\cite{Roe04}.

Finally let us consider boundary maps associated to the $\Z_2$-graded short exact sequence of 
Real $C^*$-algebras,
\[
  0 \to B \to E \xrightarrow{q}  A \to  0.
\]
The corresponding boundary map $\partial: DK(A) \to DK(B\hox \Cl_{1,0})$ was studied by van Daele~\cite{vanDaele2}.
We assume that $E$ is balanced graded and consider the composition
\[
  DK(A \hox \Cl_{0,k}) \xrightarrow{\partial} DK(B\hox \Cl_{1,k}) \xrightarrow{\calC_B} KKR( \Cl_{k,0}, B),
\]
where $\calC_B$ is the Cayley isomorphism from Theorem \ref{thm:Cayley_iso_extra_Clifford} 
 and we have made the identification 
$KKR(\Cl_{1,0}, B \hox \Cl_{1,k}) \cong KKR(\Cl_{k,0}, B)$.

Let us fix a basepoint odd self-adjoint unitary $e \in E$ and mutually anti-commuting odd self-adjoint unitaries 
$\{\gamma_j\}_{j=1}^k \subset E$ that anti-commute with $e$. These elements 
descend to anti-commuting odd self-adjoint unitaries in $\mathrm{Mult}(A)$ via the quotient map.

\begin{prop}[cf. {\cite[Section 5.2]{BKRCayley}}]  \label{prop:bdry_map_is_nice}
Let $V \in M_n(A^{\sim{q(e)}})$ be an odd self-adjoint unitary  that anti-commutes with 
$\{q(\gamma^{\oplus n}_j)\}_{j=1}^k$ and $V -q(e^{\oplus n}) \in M_n(A)$. 
Suppose that $\tilde{V} \in M_n(E)$ is a Real, odd and self-adjoint lift of $V$ 
that anti-commutes with  $\{\gamma^{\oplus n}_j\}_{j=1}^k$.
 Then the map $\calC_B \circ \partial( [V] - [e]) \in KKR(\Cl_{k,0}, B)$
can be represented by 
the bounded Real Kasparov module
\[
  \big( \Cl_{k,0}, \, B^{\oplus n}_B, \, \tilde{V} \big)
\]
with   Clifford generators $\{\gamma^{\oplus n}_j\}_{j=1}^k$.
\end{prop}
In the above result we have implicitly used that $E \subset \mathrm{Mult}(B)$ as $B$ is an ideal in $E$.
\begin{proof}
By~\cite[Proposition 5.7]{BKRCayley}, the composition $\calC_B \circ \partial$ is represented 
by the unbounded Kasparov module
\[
   \Big(\C, \, \overline{\cos (\tfrac{\pi}{2} \tilde{V}) B^{\oplus n}_B}   \hox {\Cl_{0,k}}_{\Cl_{0,k}},  \, \tan(\tfrac{\pi}{2} \tilde{V} ) \hox \one \Big).
\]
Because the operators $\{\gamma_j^{\oplus n}\}_{j=1}^k$ anti-commute with $\tilde{V}$, they give 
a well-defined representation of $\Cl_{k,0}$ on $\overline{\cos (\tfrac{\pi}{2} \tilde{V}) B^{\oplus n}_B}$ 
and anti-commute with $\tan(\tfrac{\pi}{2} \tilde{V} )$. Therefore, up to 
the isomorphism $KKR(\C, B \hox \Cl_{0,k}) \cong KKR(\Cl_{k,0}, B)$, $\calC_B \circ \partial$ can be 
represented by  the unbounded Kasparov module
\[
  \Big( \Cl_{k,0}, \, \overline{\cos (\tfrac{\pi}{2} \tilde{V}) B^{\oplus n}_B},  \, \tan(\tfrac{\pi}{2} \tilde{V} ) \Big)
\]
with Clifford generators $\{\gamma_1,\ldots,\gamma_k\}$.
We  take the bounded transform to get the Kasparov module
\[
  \Big( \Cl_{k,0}, \, B_B^{\oplus n}, \, \sin(\tfrac{\pi}{2} \tilde{V}) \Big).
\]
Finally, the straight-line operator homotopy $F_t = (1-t) \sin(\tfrac{\pi}{2} \tilde{V}) + t \tilde{V}$ for $t\in [0,1]$ does not change 
the $KKR$-class and gives the result.
\end{proof}

By the equivalence of skew-adjoint unitaries $J$ in a trivially graded 
$C^*$-algebra $A$ with odd self-adjoint unitaries $J \otimes \rho \in A \otimes \Cl_{0,1}$, 
we can also apply the results of this section to the ungraded skew-adjoint setting.

\section{Quasifree ground states from the viewpoint of SPT phases} \label{sec:Basic_Quasifree_index}

\subsection{Definition and properties}

Fermionic quasifree ground states can be naturally studied using 
Araki's self-dual CAR algebra.
Fix a separable complex Hilbert space $\mathcal{H}$ and a real involution, 
a self-adjoint anti-unitary $\Gamma$. Equivalently, $\calH$ is a Real Hilbert $\C$-module 
with real involution $v^\frakr = \Gamma v$.
The self-dual CAR algebra $A^\mathrm{car}_\mathrm{sd}(\calH,\Gamma)$ 
is the $C^*$-algebra generated by $\one$ and $\frakc(v)$ for $v\in \calH$ such that 
$v\mapsto \frakc(v)$ is linear and with relations 
$$
  \frakc(v)^* = \frakc(\Gamma v)  , \qquad \qquad \{\frakc(v)^*,\frakc(w)\}  = \langle v,w \rangle_{\calH} , 
  \qquad v, w \in \calH.
$$
The self-dual CAR algebra is $\Z_2$-graded by the parity automorphism $\Theta$, where  
$\Theta(\frakc(v))=-\frakc(v)$ for all $v\in \calH$.
One recovers the more familiar CAR algebra by means of a \emph{basis projection}, 
an orthogonal projection $P$ on $\calH$ such that $P + \Gamma P \Gamma = \one_\calH$. 
Given a basis projection, there is a graded isomorphism 
$ A^\mathrm{car}(P\calH) \cong A^\mathrm{car}_\mathrm{sd}(\calH,\Gamma)$ which on 
generators is given by 
\[
  \fraka^*(Pv)    \mapsto \frakc(Pv)  , \qquad  \qquad  \fraka(Pv)  \mapsto  \frakc(\Gamma Pv)   , 
  \qquad v \in \calH.
\]

Basis projections also are used to construct pure quasifree states on 
$A^\mathrm{car}_\mathrm{sd}(\calH,\Gamma)$. We summarise some of 
the key results of~\cite{Araki70}.

\begin{thm}[{\cite[Theorem 1]{Araki70}}] \label{thm:Araki_quasifree}
Let $P$ be a basis projection on $(\calH, \Gamma)$. 
\begin{enumerate}
  \item[{\rm (i)}] There is a quasifree, pure and $\Theta$-invariant state $\omega_P$ on $A^\mathrm{car}_\mathrm{sd}(\calH,\Gamma)$ 
such that 
$$
  \omega_P \big(\frakc(v)^* \frakc(w)\big)  =  \langle v, Pw \rangle_{\calH} , \qquad  v, \, w \in \calH,
$$
and is extended to $A^\mathrm{car}_\mathrm{sd}(\calH,\Gamma)$ by the formulas 
\begin{align*}
  &\omega_P (\frakc(v_1)\cdots \frakc(v_{2n+1})) = 0,    \\
  &\omega_P (\frakc(v_1)\cdots \frakc(v_{2n})) = (-1)^{n(n-1)/2} \sum_{\sigma} (-1)^{\sigma} 
      \prod_{j=1}^n \omega_P \big( \frakc(v_{\sigma(j)}) \frakc(v_{\sigma(j+n)}) \big) ,
\end{align*}
where $n \in \mathbb{N}$, $v_j \in \calH$ for all $j$ and the sum is over permutations $\sigma$ such that 
$$
  \sigma(1) < \sigma(2) < \ldots < \sigma(n)  , \qquad \sigma(j) < \sigma(j+n)  , \qquad  j\;=\; 1,\ldots, n .
$$
\item[{\rm (ii)}] 
Let $P_0$ and $P_1$ be basis projections on $(\calH, \Gamma)$. Then $\omega_{P_0}$ and $\omega_{P_1}$ are unitarily equivalent 
if and only if $P_0 - {P}_1$ is in the ideal of Hilbert-Schmidt operators. 
\end{enumerate}
\end{thm}

A simple method to construct quasifree states is to consider the unitary dynamics on $(\calH, \Gamma)$  
 generated by a self-adjoint operator $H=H^*$ such that 
$\Gamma \big( \Dom(H)\big) \subset \Dom(H)$ and $\Gamma H  = -H\Gamma$. We will 
call such operators Bogoliubov--de Gennes (BdG) Hamiltonians.
We will furthermore restrict to 
\emph{gapped} BdG Hamiltonians by assuming that  $0 \notin\sigma(H)$. 

A BdG Hamiltonian defines a quasifree dynamics 
$\beta: \R \to \Aut(A^\mathrm{car}_\mathrm{sd}(\calH,\Gamma))$ given on generators 
by 
$\beta_t( \frakc(v)) = \frakc(e^{itH }v)$ for all $v\in \calH$.
The ground state of this action is then completely described 
by  the basis projection $P = \chi_{(0,\infty)}(H)$, 
$P + \Gamma P \Gamma = \one_\calH$.

\begin{prop}[{\cite[Proposition 6.37]{EvansKawahigashi}}] \label{prop:gappedBdG_gappedGS}
Let $\beta:\mathbb{R} \to \mathrm{Aut}(A^\mathrm{car}_\mathrm{sd}(\calH,\Gamma))$ be a 
quasifree dynamics with BdG Hamiltonian $H$. If $0 \notin \sigma(H)$, then the 
quasifree state $\omega_{P}$ associated to the basis projection $P=\chi_{(0,\infty)}(H)$ is 
the unique ground state for the dynamics $\beta$. Furthermore, this ground state 
is gapped in the sense that the generator of the dynamics  on 
the GNS space has a spectral gap above $0$.
\end{prop}

\begin{example}[BdG Hamiltonians from superconductors] \label{ex:canonical_example}
The canonical example we will consider is the Nambu space $\calH = \calV \oplus \calV^*$ 
with $\calV$ the Hilbert space of electrons and $\calV^*$ the space of holes related by the 
(anti-linear) Riesz map $R: \calV \to \calV^*$. In particular, $\calH$ has the natural real 
involution $\Gamma = \begin{pmatrix} 0 & R^{-1} \\ R & 0 \end{pmatrix}$. We will also 
use the isomorphism $\calV \oplus \calV^* \simeq \calV \otimes \C^2$ with real involution  
$\Gamma =  \mathfrak{C} (\one \otimes \sigma_1)$ and $\mathfrak{C}$ complex conjugation.
Typical examples of $\calV$ for discrete models are $\ell^2(\Lambda, \C^n)$ with 
$\Lambda$ a countable set. For continous models we will consider $\calV = L^2(\R^d, \C^n)$ 
or $L^2(M, V)$ with $V\to M$ a complex Hermitian 
vector bundle over a complete Riemannian manifold.

Generically, we will consider BdG Hamiltonians on $\calV \otimes \C^2$ of the form 
\begin{equation} \label{eq:Generic_BdG}
  \begin{pmatrix} h & \delta \\ \delta^* & - \ol{h} \end{pmatrix}, \qquad \qquad \ol{A} = \mathfrak{C} A \mathfrak{C}, 
  \qquad \qquad \delta^* = - \ol{\delta}.
\end{equation}
In discrete systems, e.g. $\calH = \ell^2(\Z^d) \otimes \C^{2n}$, we typically have that 
$h = \big( p(S_1,\ldots, S_n) - \mu \big)\otimes \one_{n}$ with $p(S_1,\ldots, S_n)$ is a self-adjoint finite polynomial 
of the shift operators (e.g. the discrete Laplacian) and $\mu \in \R$ the Fermi energy. 
The pairing potential $\delta$ is then determined by the type of superconductor under consideration 
($s$-wave, $d$-wave $(p\pm ip)$-wave etc.). Concrete examples for $d=2$ can be found 
in~\cite[Section 2.3]{DNSB17}. 

For continuous models such as $\calH = L^2(\R^d, \C^{2n})$, the BdG Hamiltonians we consider take the same 
form as Equation \eqref{eq:Generic_BdG}, but now where $h= (\sum_j (-i\partial_j - A_j)^2  - \mu) \otimes \one_n$ 
is a (possibly magnetic) Laplacian. The coupling term $\delta$ is often a first-order differential operator and 
depends on the example under consideration, see~\cite[Section 2]{StoneRoy} for example. For the purposes 
of this paper, the specific form of $H$ is not so important provided it is sufficiently local,  
elliptic and  $0 \notin \sigma(H)$. 

Let us also note that we may also consider weakly disordered
BdG Hamiltonians, i.e. the disordered Hamiltonian still anti-commutes with $\Gamma$ and has 
a spectral gap at $0$.
\end{example}

Given 
a unitary or anti-unitary operator $W$ on $\mathcal{H}$ that commutes with $\Gamma$, we can define a
linear or anti-linear 
automorphism $\beta_W$ of $A^\mathrm{car}_\mathrm{sd}(\calH,\Gamma)$ such that 
$\beta_W( \frakc(v)) = \frakc( Wv) $. Note also that any such quasifree automorphism will 
commute with the parity automorphism $\Theta = \beta_{-\one}$.

\begin{lemma} \label{lemma:quasifree_symm_invariance}
Fix a basis projection $P + \Gamma P \Gamma = \one$ on $(\calH, \Gamma)$.
\begin{enumerate}
   \item[{\rm (i)}] If $W$ is a unitary operator on $\calH$ such that $[W,\Gamma] = 0 = [W, P]$, then
   $\omega_P ( \beta_W (a)) = \omega_P(a)$ for all $a \in A^\mathrm{car}_\mathrm{sd}(\calH,\Gamma)$.   
  \item[{\rm (ii)}] If $T$ is an anti-unitary operator on $\mathcal{H}$ such that 
$[T,\Gamma] =[T,P]= 0$, then the quasifree state $\omega_{P}$ 
is such that $ \omega_{P}( \beta_T(a^*) ) = \omega_{P}(a)$ for all $a \in A^\mathrm{car}_\mathrm{sd}(\calH,\Gamma)$.   
\end{enumerate}
\end{lemma}
\begin{proof}
Because $\omega_{P}$ is quasifree,  we only need to check invariance on operators 
of the form $\frakc(u)^* \frakc(v)$ for arbitrary $u,v \in \calH$. For the linear case, 
\begin{align*}
  \omega_{P} \big( \beta_W( \frakc(u)^* \frakc(v) )\big)  &= \omega_{P}( \frakc(Wu)^* \frakc(Wv)) 
  = \langle Wu , P Wv \rangle_{\mathcal{H}}  \\
   &= \langle Wu, W P v \rangle_{\mathcal{H}} = \langle u, P v \rangle_{\mathcal{H}} = \omega_{P}( \frakc(u)^* \frakc(v) )  .
\end{align*}
For the anti-linear case, we first note that because $(\calH,\Gamma)$ is a Real Hilbert space, $\ol{\langle u,v \rangle} = \langle \Gamma u, \Gamma v\rangle$. 
We now compute
\begin{align*}
  \omega_{P} \circ \beta_T \big(   (\frakc(u)^* \frakc(v))^* \big)  &=   \omega_{P} \big( \beta_T ( \frakc(v)^* \frakc(u) ) \big)
 = \omega_{P} ( \frakc( \Gamma T \Gamma v)^* \frakc(Tu) ) \\
  &= \langle  \Gamma T\Gamma v, P T u \rangle_{\mathcal{H}}  =\langle  Tv, PT u  \rangle_{\mathcal{H}}  
   = \langle v, Pu \rangle_\mathcal{H} \\
   &   = \ol{ \langle u, Pv \rangle_\mathcal{H} } 
   = \langle u, Pv \rangle_\calH
  = \omega_{P}( \frakc(u)^* \frakc(v) )  .    \qedhere
\end{align*}
\end{proof}

Note that the assumptions  $[T,\Gamma]=[W,\Gamma]=0$ imply that  $T$ and $W$ are (linear) orthogonal 
operators on the real Hilbert space $\calH_\R = \{ v \in \calH\,:\, \Gamma v = v\}$.

\subsection{Parity symmetry and $\Z_2$-indices}

Let us recall the basic classifying principle of 
symmetry protected topological (SPT) phases of gapped ground states with an on-site symmetry.

\vspace{0.13cm}

 {\sl Fix a reference ground state $\omega_0$ and on-site symmetry $G$. We assume that a gapped ground state 
 $\omega$ is connected to $\omega_0$ but need not be $G$-equivariantly connected.}
 
\vspace{0.13cm}

Usually $\omega_0$ is taken to be a product state.
Here we will consider a similar notion.

\vspace{0.13cm}
{\sl Let $\omega_0$ and $\omega_1$ be quasifree pure gapped ground states of $A^\mathrm{car}_\mathrm{sd}(\calH,\Gamma)$. We assume 
 that  $\omega_1$ is unitarily equivalent to $\omega_0$, but their restrictions to 
 the $\Theta$-invariant (even) subalgebra $A^\mathrm{car}_\mathrm{sd}(\calH,\Gamma)^{0}$ need not be unitarily equivalent. }
\vspace{0.13cm}

Hence, let $H_0$ and $H_1$ be BdG Hamiltonians with $0 \notin \sigma(H_0) \cup \sigma(H_1)$  and such that the 
corresponding quasifree 
ground states $\omega_0$ and $\omega_1$ are unitarily equivalent. 
By Theorem \ref{thm:Araki_quasifree}, the basis projections  $P_0 = \chi_{(0,\infty)}(H_0)$ 
and $P_1 = \chi_{(0,\infty)}(H_1)$ are such that $P_0 - P_1$ is Hilbert-Schmidt. 
Let us therefore consider the Real skew-adjoint unitaries $J_0 = i(2P_0 - \one)$, $J_1 = i(2P_1- \one)$, where 
$\Gamma J_0 \Gamma = J_0$ and $\Gamma J_1 \Gamma = J_1$. In the sequel we will often write 
$J_k = i \, \mathrm{sgn}(H_k) = i H_k |H_k|^{-1}$, $k\in \{0,1\}$, 
which is defined via the Borel functional calculus.

\begin{lemma}[{\cite[Proposition 4.4]{BCLR}}]
The sum $J_0 + J_1$ is a Real and Fredholm operator on $\calH$.
\end{lemma}

Let us therefore consider the finite-dimensional space $\Ker(J_0 + J_1)$. Noting that $J_1(J_0+J_1)=(J_0+J_1)J_0$,  
we see that both $J_0$ and $J_1$ act on $\Ker(J_0+J_1)$. Choosing $J_0$ or $J_1$ as the generator, 
we therefore obtain an ungraded $\Cl_{0,1}$-action on $\Ker(J_0+J_1)$. Note that we cannot take a $\Cl_{0,2}$-action 
 as $J_0$ and $J_1$ neither commute nor anti-commute in general.
Hence $\Ker(J_0 + J_1)$ can be regarded as an ungraded $\Cl_{0,1}$-module and we can ask 
 whether it extends to a $\Cl_{0,2}$-module. As the following result shows, this extension 
only occurs when $\omega_0$ and $\omega_1$ are equivalent on the even subalgebra.

\begin{prop}[{\cite[Theorem 4]{ArakiEvans}}]
The states $\omega_0$ and $\omega_1$ restricted to $A^\mathrm{car}_\mathrm{sd}(\calH,\Gamma)^{0}$ 
  are equivalent if and only if $\frac{1}{2} \mathrm{dim}\,\Ker(J_0+J_1)$ is even.
\end{prop}

If $H_0$ and $H_1$ are bounded, the index $\frac{1}{2} \mathrm{dim}\,\Ker(J_0+J_1)$
also computes the $\Z_2$-valued spectral flow of a skew-adjoint and Real Fredholm 
path $\{iH_t\}_{t\in[0,1]}$ connecting  $iH_0$ and $iH_1$~\cite[Proposition 6.2]{CPSB}. If $H_0$ or $H_1$ is unbounded, the 
same is true using the unbounded version of the $\Z_2$-valued spectral flow~\cite[Section 6]{BCLR}.

\subsection{Altland--Zirnbauer symmetries}

Here we consider the physical symmetries of free-fermionic systems 
as considered by Altland and Zirnbauer~\cite{AZ97, KZ16, AMZ}. 
Such symmetries may be unitary or anti-unitary operators on $\calH$ and commute with the BdG Hamiltonians. 
The quasifree ground state $\omega_P$ will then be invariant under such symmetries 
by Lemma \ref{lemma:quasifree_symm_invariance}.
In particular, we will take advantage of the following result first noted by 
Kennedy and Zirnbauer and then further developed by Alldridge, Max and Zirnbauer.

\begin{prop}[{\cite[Section 2]{KZ16}, \cite[Section 3.1--3.2]{AMZ}}] \label{prop:KZ_symmetries}
Let $H$ be a BdG Hamiltonian on $(\calH,\Gamma)$ with $0 \notin \sigma(H)$ and $J= iH|H|^{-1}$.
If $H$ has Altland--Zirnbauer symmetries, then there are 
Real mutually anti-commuting skew-adjoint unitaries $\{\kappa_j\}_{j=1}^n \subset \calB(\calH)$ 
such that $\kappa_j J  = -  J \kappa_j$ for all $j \in \{1,\ldots, n\}$ and where the integer  $n \in \{0,\ldots, 7\}$ depends 
on the symmetry.
\end{prop}

Proposition \ref{prop:KZ_symmetries} as stated hides many details, but we remark that the 
skew-adjoint unitaries $\{\kappa_j\}_{j=1}^n$ are concretely constructed from the physical symmetry 
operators that commute with $H$. For example, if $H$ is time-reversal symmetric via a self-adjoint anti-unitary 
$T$ such that $[T,H] = [T,\Gamma]=0$ and $T^2 = -\one$, we take $\kappa_T = \Gamma T$.

We consider gapped BdG Hamiltonians $H_0$ and $H_1$ which have the same Altland--Zirnbauer symmetry 
type and whose gapped ground states $\omega_0$ and $\omega_1$ are unitarily 
equivalent if we ignore symmetries.
That is, $J_0 - J_1$ is Hilbert-Schmidt and there are operators $\{\kappa_j\}_{j=1}^n$ such that 
$J_i \kappa_j = - \kappa_j J_i$ for all $i \in \{0,1\}$, $j \in \{1,\ldots,n\}$.

Therefore we again consider  the finite-dimensional space $\Ker(J_0+J_1)$, where $\kappa_j \cdot \Ker(J_0 + J_1) \subset \Ker(J_0+J_1)$.
Hence, $\Ker(J_0+J_1)$ is an ungraded $\Cl_{0, n+1}$-module with generators $\{J_0, \kappa_1,\ldots,\kappa_n\}$. 
 Letting $\calM_n$ denote the Grothendieck group of equivalence classes of ungraded $\Cl_{0,n}$-modules, we can 
 consider the class $[\Ker(J_0+J_1)] \in \calM_{n+1}/\calM_{n+2} \cong KO_{n+2}(\R)$ via the 
Atiyah--Bott--Shapiro isomorphism~\cite[Theorem 11.5]{ABS}. 
The vanishing of this class implies that the $KO_{n+2}(\R)$-valued spectral flow between the (bounded or unbounded) 
skew-adjoint endpoints
$iH_0$ and $iH_1$ vanishes~\cite[Section 5--6]{BCLR}. A non-trivial Clifford index guarantees that the ground state gap will close 
on any Fredholm path connecting 
$\omega_0$ and $\omega_1$. We therefore say that $[\Ker(J_0+J_1) ] \in KO_{n+1}(\R)$ is a topological obstruction to connect 
the two ground states $\omega_0$ and $\omega_1$ in a way that respects the Altland--Zirnbauer symmetries of $H_0$ and $H_1$.

\begin{remark}[Reinterpretation via Cayley map]
Let $J_0$ and $J_1$ be Real skew-adjoint unitaries on $(\calH,\Gamma)$ that anti-commute with 
the ungraded $\Cl_{0,n}$-generators $\{\kappa_j\}_{j=1}^n$. Suppose further that 
and $J_0 - J_1$ is Hilbert-Schmidt. In particular, $\|J_0 - J_1\|_{\calQ(\calH)} = 0$ and so the operators 
$J_0, \,J_1 \in \mathrm{Mult}(\mathbb{K}(\calH))$ fall into the framework of Section \ref{subsec:skew_cayley}. 
Hence there is an unbounded Kasparov module
\[
  \Big( \Cl_{n+2,0}, \, \ol{(J_0-J_1)\mathbb{K}}_\mathbb{K} \otimes \bigwedge\nolimits^{\! *} \C, \, J_0(J_1+J_0)(J_1 - J_0) \otimes \rho \Big), \quad \mathbb{K} = \mathbb{K}(\calH)
\]
with left Clifford generators $\{\one\otimes \gamma, J_0\otimes \rho, \kappa_1\otimes \rho, \ldots, \kappa_n\otimes \rho\}$.
We therefore obtain a class in $KKR(\Cl_{n+2,0}, \mathbb{K}(\calH)) \cong KO_{n+2}(\mathbb{K}(\calH)) \cong KO_{n+2}(\R)$ which can be 
determined by the ungraded Clifford module index $[\Ker(J_0+J_1)] \in KO_{n+2}(\R)$.
\end{remark}

\subsection{Compact group symmetries}

Let $G$ be a compact group and $\nu:G \to \Z_2$ a homomorphism. 
A representation $W$ of $G$ is  unitary/anti-unitary with respect to $\nu$ if $W_g$  is unitary (respectively anti-unitary)
for all $g \in G$ such that  $\nu(g)=0$ (respectively $\nu(g)=1$).  
We fix a unitary/anti-unitary representation $W$ and assume that $W_g \Gamma = \Gamma W_g$ for all $g \in G$. 
This then gives a linear/anti-linear action $\beta$ on $A^\mathrm{car}_\mathrm{sd}(\calH,\Gamma)$ 
relative to $\nu: G\to \Z_2$
such that $\beta_g(\frakc(v)) = \frakc( W_g v)$ for any $g \in G$ and $v \in \calH$.
Now take a gapped BdG Hamiltonian such that $[H,W_g]$ is well-defined and $[H,W_g]=0$ for all $g \in G$. 
Then for $P = \chi_{(0,\infty)}(H)$ the basis projection on $(\calH,\Gamma)$,  $[W_g, P]=0$ for any $g\in G$ and 
the ground state $\omega_P$ is invariant under $\beta$ by Lemma \ref{lemma:quasifree_symm_invariance}.
For $J = i \,\mathrm{sgn}(H) = iH|H|^{-1}$, a simple computation shows that
\[
   W_g J W_g^* = (-1)^{\nu(g)} J \quad \text{ for all } g\in G.
\]

Following the viewpoint of SPT phases, we are also interested in  ground states on the $G$-invariant 
subalgebra 
\[
A^\mathrm{car}_{\mathrm{sd}}(\calH,\Gamma)^{G}  = \{ a\in A^\mathrm{car}_{\mathrm{sd}}(\calH,\Gamma) \, : \,
 \beta_g(a) = a^{\nu(g)\ast} \text{ for all } g\in G \}, 
 \quad a^{\nu(g)\ast} = \begin{cases} a, & \nu(g) = 0 \\ a^*, & \nu(g)=1 \end{cases}.
\]
Given a state $\omega$ on $A^\mathrm{car}_{\mathrm{sd}}(\calH,\Gamma)$, we let 
$\omega^G$ denote its restriction to $A^\mathrm{car}_{\mathrm{sd}}(\calH,\Gamma)^G$.

Let us now take two $G$-symmetric gapped BdG Hamiltonians $H_0$ and $H_1$ such that 
$\omega_0$ and $\omega_1$ are equivalent without symmetry. 
Then the space $\Ker(J_0+J_1)$ is finite dimensional and $W_g \cdot\Ker(J_0+J_1) \subset \Ker(J_0+J_1)$ 
for all $g \in G$. 
Therefore $\Ker(J_0 + J_1)$ determines  an element of $R(G)$, the representation ring of $G$.

\begin{defn}[{cf. \cite[Section 5]{CareyEvans}}]
We denote by $R(G)_k$ the Grothendieck group of ungraded and real $(\Cl_{0,k-1}, G)$-bimodules modulo those 
extendable to ungraded and real $(\Cl_{0,k}, G)$-bimodules.
\end{defn}

By a $G$-equivariant extension of the Atiyah--Bott--Shapiro isomorphism, 
$R(G)_k \cong KO^G_{k}(\R)$~\cite[Section 2]{KasparovNovikov}. 
The ground states  $\omega_0$ and $\omega_1$ are such that $\Ker(J_0+J_1)$ is a 
$(\Cl_{0,1}, G)$-module with the Clifford generator $J_0$ and the unitary/anti-unitary representation $W$.  
We can therefore associate the element $[\Ker(J_0+J_1)] \in R(G)_2$.

\begin{thm}[{\cite[Theorem A]{Matsui87}, \cite[Theorem 5.1]{CareyEvans}}]
Let $G$ be a compact group and $W$ a unitary/anti-unitary representation on  $(\calH,\Gamma)$ 
commuting with $\Gamma$. 
Suppose that $P_0$ and $P_1$ are basis projections on $(\calH,\Gamma)$ such that for all $g\in G$,
$[W_g, P_0] = [W_g, P_1]=0$. Then 
the states $\omega_{P_0}^G$ and $\omega_{P_1}^G$ 
   on $A^\mathrm{car}_\mathrm{sd}(\mathcal{H},\Gamma)^G$  
   are equivalent if and only if:
 \begin{enumerate}
  \item[{\rm (i)}] $P_0 - P_1$ is a Hilbert-Schmidt operator, 
  \item[{\rm (ii)}] $[\Ker(J_0+J_1)]$ represents the identity in $R(G)_2$ 
  with $J_k = i(2P_k -\one)$ for $k\in \{0,1\}$.
\end{enumerate}
\end{thm}

The result in~\cite{Matsui87} is stated for unitary representations, but the proof also holds 
for unitary/anti-unitary actions commuting with $\Gamma$ as $W$ restricts to a (linear) orthogonal 
representation on  $\calH_\R = \{ v\in \calH\,:\, \Gamma v = v \}$.

\section{Locality and $KO_{\ast}(A^\frakr)$-valued indices} \label{sec:KO_index_general}

Let us now move to a more abstract setting, where we have a $C^{\ast,\frakr}$-algebra 
$(A,\frakr)$ such that $(\mathrm{Mult}(A), \frakr) \subset (\calB(\calH),\mathrm{Ad}_\Gamma)$.  In 
particular, $a^\frakr = \Gamma a \Gamma$ for all $a \in \mathrm{Mult}(A)$. 
When $A= \mathbb{K}(\calH)$, we recover the results of the previous section.
We will assume that $A$ is trivially graded, 
but will often consider tensor products of the form $A \otimes \Cl_{r,s}$, which is $\Z_2$-graded. 
We assume that $A$ is $\sigma$-unital but make no further  assumptions on the specific form of $A$, so $A$ may 
be a crossed product, groupoid or (uniform) Roe algebra. In Section \ref{sec:Coarse_and_KHom} 
we will consider the case  $A = C^*(X)$.
We also note that the content of this section is similar to  previous work by Alldridge, Max and Zirnbauer~\cite{AMZ}.

\subsection{Basic construction}  \label{Subsec:General_KTheory_class_noSymm}

Because we also  consider unbounded BdG Hamiltonians, we recall the notion of a 
normalising function.
\begin{defn}
We call a continuous odd function $\chi: \R \to [-1,1]$ a normalising function if 
$\chi(t) \to \pm 1$ as $x \to \pm \infty$.
\end{defn}

We fix a Hilbert space $\calH$ and real involution $\Gamma$. 
We wish to consider topological properties of BdG Hamiltonians and quasifree ground states 
 relative to the $C^*$-algebra $A$. We therefore make the following assumption for this section.

\begin{assumption} \label{assump:A_locality}
Let $H_0$ and $H_1$ act on $(\calH,\Gamma)$ such  that for $k\in\{0,1\}$, 
$0 \notin \sigma(H_k)$, and $H_k=H_k^*=-\Gamma H_k \Gamma$. 
\begin{enumerate}
  \item[(a)] (Unbounded case) If $H_0$ and $H_1$ are unbounded, we assume that 
  $\chi(H_0), \, \chi(H_1) \in \mathrm{Mult}(A)$ and $\chi(H_0)-\chi(H_1) \in A$ for any normalising function $\chi$.
  \item[(b)] (Bounded case) If $H_0$ and $H_1$ are bounded, we assume that 
  $H_0$ and $H_1$ are invertible elements of $\mathrm{Mult}(A)$ and $H_0- H_1 \in A$.
\end{enumerate}
\end{assumption}

Physically, Assumption \ref{assump:A_locality} implies that we are given some Real $C^*$-algebra $A$
that specifies the allowed deformations of the BdG Hamiltonians within the Nambu space $(\calH,\Gamma)$.

\begin{lemma}
\begin{enumerate}
  \item[{\rm (i)}] Given Assumption \ref{assump:A_locality}(a), $J_0=iH_0|H_0|^{-1}$ and $J_1 = iH_1 |H_1|^{-1}$ are Real invertible elements of $\mathrm{Mult}(A)$. 
Furthermore, $J_0 - J_1$ and $P_0 -P_1 \in A$ for $P_0 = \chi_{(0,\infty)}(H_0)$ and $P_1 = \chi_{(0,\infty)}(H_1)$.
  \item[{\rm (ii)}] Given Assumption \ref{assump:A_locality}(b), then $\chi(H_0), \, \chi(H_1) \in \mathrm{Mult}(A)$
  for any normalising function $\chi$ and  $J_0 - J_1 \in A$ 
\end{enumerate}
\end{lemma}
\begin{proof}
(i) For a fixed $k\in \{0,1\}$, by assumption there is some $\epsilon>0$ such that $[-\epsilon,\epsilon] \cap \sigma(H_k) = \emptyset$. 
Therefore there is a 
suitable normalising function $\chi$ such that $i\chi(H_k) = J_k$. It is then immediate that $J_k \in \mathrm{Mult}(A)$ for 
any $k \in \{0,1\}$ and $J_0-J_1 \in A$. The results for $P_0$ and $P_1$ immediately follow as $J_k = i(2P_k-\one)$, $k \in \{0,1\}$.

(ii) Because $H_0$ and $H_1$ are bounded and invertible,
the first statement immediately follows from the continuous functional calculus. 
Next we write $H_1 = H_0+a$ for some $a \in A$ and note that 
\[
  H_0|H_0|^{-1} - H_1|H_1|^{-1} = H_0 \big( |H_0|^{-1} - |H_1|^{-1}\big) + a |H_1|^{-1}.
\]
Because $H_1$ is invertible in $\mathrm{Mult}(A)$, $|H_1|^{-1} \in \mathrm{Mult}(A)$ and so $a |H_1|^{-1} \in A$. 
The result will therefore follow if we can show that $|H_0|^{-1} - |H_1|^{-1} \in A$. We will use the 
integral formula for fractional powers, where for any self-adjoint and invertible operator $T$
\[
  |T|^{-1} = (T^2)^{-1/2} = \frac{2}{\pi} \int_0^\infty (T^2 + x^2)^{-1} \,\mathrm{d}x.
\]
Using the resolvent identity, we therefore find that 
\begin{align*}
   |H_0|^{-1} - |H_1|^{-1} &= \frac{2}{\pi} \int_0^\infty \Big( (H_0^2+x^2)^{-1} -  ( (H_0+a)^2 +x^2)^{-1} \Big) \mathrm{d}x \\
   &= \frac{2}{\pi} \int_0^\infty \Big( (H_0^2+x^2)^{-1} ((H_0 +a)^2 - H_0^2) ( (H_0+a)^2 +x^2)^{-1} \Big) \mathrm{d}x.
\end{align*}
Because $(H_0 +a)^2 - H_0^2 = a^2+aH_0 + H_0a \in A$, the integrand is an element of $A$. Finally for any 
$b \in A$, the continuous functional calculus gives the estimate
\[
   \big\| (H_0^2+x^2)^{-1}b ((H_0+a)^2+x^2)^{-1} \big\|_A \leq C (1+x^2)^{-1},
\]
which then implies that the integral will be norm-convergent in $A$.
\end{proof}

The equivalence condition on the gapped ground states $\omega_0$ and $\omega_1$ is now 
turned into a locality-like condition, where we assume that while $J_0$ and $J_1$ are 
in $\mathrm{Mult}(A)$, their difference $J_0-J_1 \in A$.

We note that 
$J_0 \otimes \rho, J_1 \otimes \rho \in \mathrm{Mult}(A) \otimes \Cl_{0,1}$ are Real odd self-adjoint unitaries. 
Recalling Section \ref{subsec:vanDaele_and_boundary}, if $J_0-J_1 \in A$, 
then we obtain a class in the relative van Daele $K$-theory group
\begin{equation} \label{eq:vD_class_general}
  \big[ J_1 \otimes \rho\big] - \big[ J_0 \otimes \rho \big] \in DK(\mathrm{Mult}(A)\otimes \Cl_{0,1}, \, \mathrm{Mult}(A)/ A \otimes \Cl_{0,1}) \cong DK(A \otimes \Cl_{0,1}).
\end{equation}
Because $A$ is ungraded, $DK(A \otimes \Cl_{0,1}) \cong KO_2( A^\frakr)$.

Considering $J_0, \, J_1 \in \mathrm{Mult}(A)$ as operators on the Hilbert $C^*$-module $A_A$, if $J_0-J_1 \in A$, 
then $\|J_0-J_1\|_{\calQ_A} = 0$ and 
we can  apply Proposition \ref{prop:skew_cayley_kk} to obtain the Real unbounded Kasparov module
\begin{equation} \label{eq:Basic_Cayley_class}
  \Big( \Cl_{2,0}, \, \ol{(J_0-J_1)A}_A \otimes \bigwedge\nolimits^{\! *} \C, \,  \calC_{J_0}(J_1) \otimes \rho \Big)
\end{equation}
with $\Cl_{2,0}$-generators $\{\one \otimes \gamma, J_0\otimes \rho \}$. We denote the corresponding 
class as $\big[\calC_{J_0}(J_1) \big] \in KKR(\Cl_{2,0}, A) \cong KO_2(A^\frakr)$. 

By Proposition \ref{prop:skew_cayley_kk}, this $KKR$-theory class can also be represented by the bounded 
Kasparov module
\begin{equation} \label{eq:Bdd_transform_on_whole_module}
   \big[\calC_{J_0}(J_1) \big] = \Big[ \Big( \Cl_{n+2,0}, \,  A_A \otimes \bigwedge\nolimits^{\! *} \C, \,  F_{\calC_{J_0}(J_1)} \otimes \rho \Big) \Big] 
   \in KKR( \Cl_{2,0}, A)
\end{equation}
with $F_{\calC_{J_0}(J_1)} = \calC_{J_0}(J_1)(\one - \calC_{J_0}(J_1)^2)^{-1/2}$ the bounded transform.

The following 
result is clear to readers familiar with Kasparov theory, but we state it with a proof for completeness.

\begin{prop} \label{prop:Connected_J_trivial_K}
Suppose that $\{J_t\}_{t\in [0,1]}$ is a continuous path of skew-adjoint Real unitaries in $\mathrm{Mult}(A)$ such 
that $J_t - J_0 \in A$ for all $t\in [0,1]$. Then the corresponding class $\big[ \calC_{J_0}(J_1) \big]$ is trivial 
in $KKR(\C, A\otimes \Cl_{0,2})$.
\end{prop}
\begin{proof}
The conditions on the path $\{J_t\}_{t\in [0,1]}$ imply that the pointwise Hilbert $C^*$-module 
$\ol{(J_t - J_0)A}_A$ can be completed to a $A\otimes C([0,1])$-module, $\ol{(J_\bullet -J_0)(A\otimes C([0,1])}_{A\otimes C([0,1])}$. 
The continuity of $J_t$ also ensures that the pointwise bounded transform 
 $F_{\calC_{J_0}(J_t)} = -\tfrac{1}{2} J_0(J_tJ_0-1)(J_0J_t+1)^{-1}(2+J_0J_t+J_tJ_0)^{1/2}$ 
gives a well-defined and skew-adjoint and Fredholm operator $F_\bullet$ on 
$\ol{(J_\bullet -J_0)(A\otimes C([0,1])}_{A\otimes C([0,1])}$. Hence the triple
\[
   \Big( \Cl_{2,0}, \, \ol{(J_\bullet -J_0)(A\otimes C([0,1])}_{A\otimes C([0,1])} \otimes \bigwedge\nolimits^{\! *} \C, \,  F_\bullet \otimes \rho \Big)
\]
is a well-defined Kasparov module that gives a homotopy in $KKR(\Cl_{2,0}, A)$. However, evaluating at $t=0$, the 
corresponding Kasparov module $\big( \Cl_{2,0}, \, 0_A, \, 0\big)$ is degenerate and therefore 
trivial in $KKR(\Cl_{2,0}, A)$. Hence $\big[ \calC_{J_0}(J_1) \big] = \big[ \calC_{J_0}(J_0) \big]$ 
is also trivial.
\end{proof}

The contrapositive of Proposition \ref{prop:Connected_J_trivial_K} says that if $\big[ \calC_{J_0}(J_1) \big]$ is non-trivial in $KO_2(A^\frakr)$, 
then $J_0$ and $J_1$  cannot be connected by a path $\{J_t\}_{t\in [0,1]}$ that is local with respect to the algebra $A$, 
$J_t - J_0 \in A$ for all $t \in [0,1]$. 

We now relate our Cayley Kasparov module to the relative van Daele $K$-theory class 
$[J_0\otimes \rho]-[J_1\otimes \rho] \in DK(A \otimes \Cl_{0,1})$. 
\begin{prop} \label{prop:vD_class_to_Cayley_class}
The Cayley isomorphism $\calC: DK(A \otimes \Cl_{0,1}) \xrightarrow{\simeq} KKR( \Cl_{2,0}, A )$ from Theorem 
\ref{thm:Cayley_iso_extra_Clifford} is such that 
\[
   \calC \big( [J_1\otimes \rho] - [J_0\otimes \rho]) = \big[ \calC_{J_0}(J_1) \big].
\]
\end{prop}
\begin{proof}
The Cayley isomorphism maps $[J_1\otimes \rho] - [J_0\otimes \rho]$ to the class of
\begin{align*}
  \Big( \Cl_{1,0}, \, \ol{(J_0-J_1)A}_{A} \otimes {\Cl_{0,1}}_{\Cl_{0,1}},\, \calC_{J_0}(J_1)\otimes \rho \Big) 
   \sim \Big( \Cl_{2,0}, \, \ol{(J_0-J_1)A}_{A} \otimes \bigwedge\nolimits^{\!*} \C, \, \calC_{J_0}(J_1) \otimes \rho\Big).
\end{align*}
Hence we recover $\big[ \calC_{J_0}(J_1) \big]$.
\end{proof}

Proposition \ref{prop:vD_class_to_Cayley_class} also gives a simpler proof of Proposition 
\ref{prop:Connected_J_trivial_K} as a continous path of Real skew-adjoint unitaries  
$\{J_t\}_{t\in [0,1]}\subset \mathrm{Mult}(A)$ with $J_t -J_0 \in A$ for all $t \in [0,1]$ 
immediately implies that the
relative van Daele class $[J_0\otimes \rho] - [J_1 \otimes \rho]$ is trivial.

\begin{remark}[Altland--Zirnbauer symmetries relative to $A$] \label{rk:AZ_class_general}
Recalling Proposition \ref{prop:KZ_symmetries}, let us furthermore 
assume that there exists a family of mutually anti-commuting, skew-adjoint and 
Real unitaries $\{\kappa_j\}_{j=1}^n \subset \mathrm{Mult}(A)$ such that 
$J_k \kappa_j = - \kappa_j J_k$ for $j \in \{1,\ldots, n\}$ and $k \in \{0,1\}$.

We can again apply Proposition \ref{prop:skew_cayley_kk} and obtain the Real unbounded Kasparov module
\[
  \Big( \Cl_{n+2,0}, \, \ol{(J_0-J_1)A}_A \otimes \bigwedge\nolimits^{\! *} \C, \, \calC_{J_0}(J_1)\otimes \rho  \Big)
\]
with Clifford generators  $\{\one \otimes \gamma, J_0 \otimes \rho, \kappa_1\otimes \rho, \ldots, \kappa_n \otimes \rho\}$, 
which  gives a class  $\big[ \calC_{J_0}(J_1) \big] \in KKR(\Cl_{n+2,0}, A) \cong KO_{n+2}(A^\frakr)$.

As in the case without symmetry, if the class $\big[ \calC_{J_0}(J_1) \big]$ is non-zero in $KO_{n+2}(A^\frakr)$, then 
$J_0$ and $J_1$ cannot be connected by a path $\{J_t\}_{t \in [0,1]}$ that anti-commutes with 
$\{\kappa_j\}_{j=1}^n$ and is such that $J_t - J_0 \in A$ for all $t \in [0,1]$.
Hence the class $\big[ \calC_{J_0}(J_1) \big]$ represents a topological obstruction to connect the two symmetric 
BdG Hamiltonians via a path that respects the Altland--Zirnbauer symmetries and is local with respect to the 
auxiliary algebra $A$.

We can similarly consider the relative van Daele class, where the odd self-adjoint unitaries 
$J_0 \otimes \rho$ and $J_1 \otimes \rho$ in $\mathrm{Mult}(A) \otimes \Cl_{0,1}$ 
anti-commute with $\{\kappa_j \otimes \rho\}_{j=1}^n$. 
Recalling Remark \ref{rk:extra_clifford_gens_shifts_the_degree}, we have that 
\[
  [J_0 \otimes \rho] - [J_1 \otimes \rho] \in DK(\mathrm{Mult}(A) \otimes \Cl_{0, n+1}, \, \mathrm{Mult}(A) \otimes \Cl_{0, n+1}/ A \otimes \Cl_{0,n+1} ).
\]
and the same basic argument as Propositon \ref{prop:vD_class_to_Cayley_class} gives that this relative class 
is represented by $\big[ \calC_{J_0}(J_1) \big]$ under the isomorphism 
$DK(A \otimes \Cl_{0, n+1}) \xrightarrow{\simeq} KKR( \Cl_{n+2,0}, A)$.
\end{remark}

\subsection{Compact $G$-symmetries} \label{subsec:G_equiv_index_general}

Let us fix a compact group $G$ and consider symmetries  via a linear/anti-linear action $\beta$ of 
$G$ on $A$ relative to $\nu:G\to\Z_2$. 
That is, for a given $g \in G$, $\beta_g$ is linear when $\nu(g)=0$ and anti-linear when $\nu(g)=1$. Such an action has a unique 
extension to $\mathrm{Mult}(A)$, which we also denote by $\beta$. We  assume 
that this action is compatible with the real structure, 
$\beta_g(a^\frakr) = [\beta_g(a)]^\frakr$ for all $a\in \mathrm{Mult}(A)$ and $g\in G$. 
Such group actions may be built from a unitary/anti-unitary 
representation $W$ of $G$ on $(\calH,\Gamma)$ relative to $\nu:G\to\Z_2$ and 
such that $\mathrm{Ad}_{W_g}(\mathrm{Mult}(A)) \subset \mathrm{Mult}(A)$ and $[W_g,\Gamma]=0$ for all $g\in G$.

Let us now consider BdG  Hamiltonians $H_k$, $k\in \{0,1\}$, satisfying Assumption \ref{assump:A_locality}. 
If the BdG Hamiltonians satisfy Assumption \ref{assump:A_locality}(b), then we furthermore assume
 that 
$\beta_g(H_k) = H_k$ for all $g\in G$ and $k\in \{0,1\}$. This then implies that $\beta_g(J_k) = (-1)^{\nu(g)} J_k$ 
for all $g \in G$, $k\in \{0,1\}$ and $J_k = iH_k |H_k|^{-1}$.
If the BdG Hamiltonians satisfy Assumption \ref{assump:A_locality}(a), then we assume that
$\beta_g(J_k) = (-1)^{\nu(g)} J_k$ for all $g \in G$ and $k\in \{0,1\}$.

Given a $C^*$-algebra $B$ with $G$-action $\beta: G\to \Aut(B)$, recall that 
a Hilbert $C^*$-module $E_B$ is $G$-equivariant if there is an action $\wt{\eta}$ of $G$ on $E$ such that
\[
  \wt{\eta}_g( e \cdot b) = \wt{\eta}_g(e) \cdot \beta_g(b), \qquad \qquad  ( \wt{\eta}_g(e_1) \mid \wt{\eta}_g(e_2) )_B = \beta_g\big( ( e_1 \mid e_2)_B \big)
\]
for all $e, e_1,e_2 \in E_B$, $b \in B$ and $g\in G$.
A linear/anti-linear action on $E_B$ extends to a linear/anti-linear 
action ${\eta}$ on $\End_B(E)$, where for $T \in \End_B(E)$, ${\eta}_g(T)e = \wt{\eta}_g\big( T \wt{\eta}_{g^{-1}}(e)\big)$ 
for all $e \in E_B$ and $g \in G$.
The following is a simple check.
\begin{lemma}
Given $J_0, \, J_1 \in \mathrm{Mult}(A)$ and a linear/anti-linear group action $\beta$ on $\mathrm{Mult}(A)$ 
relative to $\nu: G\to \Z_2$ as above, 
the Hilbert $C^*$-module $\ol{(J_0-J_1)A}_A$ is $G$-equivariant via the action 
$\wt{\beta}_g( (J_0-J_1)a) = (-1)^{\nu(g)}(J_0-J_1)\beta_g(a)$ for $a \in A$ and $g \in G$.
\end{lemma}

We note that $\Dom(\calC_{J_0}(J_1))$ is invariant under $\wt{\beta}_g$  and  $\beta_{g}(\calC_{J_0}(J_1)) = (-1)^{\nu(g)} \calC_{J_0}(J_1)$.
In order to construct a $G$-equivariant Kasparov module, we would like $\calC_{J_0}(J_1) \otimes \rho$ to be $G$-invariant. 
We therefore define 
the following linear/anti-linear action $\wt{\alpha}$  on $\ol{(J_0-J_1)A}_A \otimes \bigwedge^*\C$,
\begin{equation} \label{eq:G_action_for_J0J1_kasmod}
  \wt{\alpha}_g \big( (J_0-J_1)a \otimes w\big) = (-1)^{\nu(g)} (J_0-J_1) \beta_g (a) \otimes \gamma^{\nu(g)} w 
\end{equation}
for any $a \in A$, $w \in \bigwedge^*\C$ and $g \in G$.
We therefore see that under the induced action of   $\wt{\alpha}$, 
\begin{align*}
\wt{\alpha}_g \circ (\calC_{J_0}(J_1) \otimes \rho) \circ \wt{\alpha}_{g^{-1}} \big( (J_0-J_1)a \otimes w \big) 
&= (-1)^{\nu(g)} J_0(J_1+J_0) (\beta_{g}\circ \beta_{g^{-1}})(a) \otimes \gamma^{\nu(g)} \rho \gamma^{\nu(g)} w \\
&= J_0(J_1+J_0) a \otimes w \\
&= (\calC_{J_0}(J_1) \otimes \rho)( (J_1-J_1)a \otimes w)
\end{align*}
for any $g \in G$, $a \in A$ and $w \in \bigwedge^*\C$. Hence
 $\calC_{J_0}(J_1) \otimes \rho$ is $G$-invariant under this action.

One can similarly check that the left $\Cl_{2,0}$-generators on $\ol{(J_0-J_1)A}_A \otimes \bigwedge^*\C$ are invariant 
under the induced action on $\alpha$ on $\End_A \big( \ol{(J_0-J_1)A}  \otimes \bigwedge^* \C\big)$ 
from $\wt{\alpha}$,  
\[
  \alpha_g( \one \otimes \gamma ) = \one \otimes \gamma, \qquad \qquad
  \alpha_g( J_0 \otimes \rho) = (-1)^{\nu(g)}J_0 \otimes (-1)^{\nu(g)} \rho = J_0\otimes \rho,
\]
for any $g \in G$.
We summarise our results.

\begin{prop} \label{prop:Equiv_K_theory_class_general}
Let $H_0$ and $H_1$ be BdG Hamiltonians satisfying Assumption \ref{assump:A_locality} 
and $J_k = i\,\mathrm{sgn}(H_k)$, $k\in \{0,1\}$. Let $\nu: G\to \Z_2$ be a 
homomorphism and suppose that 
there is a linear/anti-linear action $\beta$ of $G$ on $\mathrm{Mult}(A)$ relative to $\nu$ and 
such that $\beta_g(J_k) = (-1)^{\nu(g)}J_k$ for all $g \in G$ and $k \in \{0,1\}$. Then 
for $\wt{\alpha}$ the linear/anti-linear group action on the Hilbert $C^*$-module given by 
Equation \eqref{eq:G_action_for_J0J1_kasmod}, the triple
\[
  \Big( \Cl_{2,0}, \, \ol{(J_0-J_1)A}_A \otimes \bigwedge\nolimits^{\! *} \C, \, \calC_{J_0}(J_1) \otimes \rho \Big)
\]
is a Real unbounded  Kasparov module with Clifford generators $\{\one \otimes \gamma, J_0 \otimes \rho\}$ that
that restricts to a real $G$-equivariant unbounded Kasparov module on the elements fixed by the real structure.
\end{prop}

Once again, the class $\big[\calC_{J_0}(J_1)\big] \in KO_2^G(A^\frakr)$ of the unbounded 
Real Kasparov module from Proposition \ref{prop:Equiv_K_theory_class_general} gives an obstruction to the existence 
of a $G$-invariant path of Real skew-adjoint unitaries $\{J_t\}_{t\in [0,1]} \subset \mathrm{Mult}(A)$ such that $J_t-J_0 \in A$. A 
bounded representative of $\big[\calC_{J_0}(J_1)\big]$ can be written as 
\[
  \Big( \Cl_{2,0}, \, A_A \otimes \bigwedge\nolimits^{\! *} \C, \, F_{\calC_{J_0}(J_1)} \otimes \rho \Big)
\]
with with Clifford generators $\{\one \otimes \gamma, J_0 \otimes \rho\}$ and 
group action $\wt{\alpha}_g(a \otimes w) = \beta_g(a) \otimes \gamma^{\nu(g)}w$ for all 
$a \otimes w \in A \otimes \bigwedge^*\C$ and $g \in G$.

\subsection{Defects and short exact sequences} \label{sec:DefectsAndSES}

Many defects  of interest in condensed matter physics such as  codimension $1$ boundaries, screw-dislocations  etc. 
have a description via a  short exact sequence of observable $C^*$-algebras.  With these examples in mind, 
we assume that there is a semi-split short exact sequence of $C^{\ast,\frakr}$-algebras
\begin{equation} \label{eq:Bulk_Defect_Extension}
  0 \to B \to E \xrightarrow{\phi} A \to 0.
\end{equation}
We will assume that $E$ is unital. Using that $\mathrm{Mult}(A) = \End_A(A)$, the map $\phi:E \to A$ induces
a unital $\ast$-homomorphism (also denoted by $\phi$), $\phi :E \to \mathrm{Mult}(A)$. 
Given skew-adjoint unitaries $J_0, \, J_1 \in \mathrm{Mult}(A)$, we will assume 
a triviality condition on $J_0$, namely that a lift $T \in \phi^{-1}(J_0)$ is invertible. Therefore, 
the phase $\wt{J}_0 = T|T|^{-1}$ is a Real skew-adjoint unitary in $E$ 
that will act as a basepoint. We also fix a skew-adjoint lift $\wt{J}_1 \in E$ 
of $J_1$.

Recalling Proposition \ref{prop:KZ_symmetries}, let us also assume that we have Altland--Zirnbauer 
symmetries that are compatible with respect the short exact sequence of 
Equation \eqref{eq:Bulk_Defect_Extension}. That is, we assume that there are 
mutually anti-commuting skew-adjoint unitaries $\{\kappa_j\}_{j=1}^n \subset E$ 
such that 
\[
  \kappa_j = \kappa_j^\frakr  = - \kappa_j^*, \qquad \qquad
  \kappa_j \wt{J}_0 = -\wt{J}_0 \kappa_j,  \qquad \qquad   \kappa_j \wt{J}_1 = -\wt{J}_1 \kappa_j
\]
for all $j \in \{1,\ldots, n\}$.
Such a condition ensures that $\{\phi(\kappa_j)\}_{j=1}^n \subset \mathrm{Mult}(A)$ act as Altland--Zirnbauer  symmetries 
in the sense of Remark \ref{rk:AZ_class_general}.

A simple way to obtain such Altland--Zirnbauer symmetries is if there is a vector space 
$\calW$ such that $\{ \kappa_j\}_{j=1}^n \subset \End(\calW)$ and there is a factorisation of the 
short exact sequence 
\[
  0 \to B' \otimes \End(\calW) \to E' \otimes \End(\calW) \to A' \otimes \End(\calW) \to 0.
\]

Assuming the condition that $J_0-J_1 \in A$, we can construct a bounded or unbounded Kasparov 
module that gives a class in $KKR(\Cl_{n+2}, A)$. Our task is to examine the image of this class under 
the boundary map induced from Equation \eqref{eq:Bulk_Defect_Extension}. 

\begin{prop} \label{prop:Boundary_KK_class}
The image of $\big[ \calC_{J_0}(J_1)\big]$  under the 
boundary map $\partial: KKR( \Cl_{n+2,0}, A) \to KKR( \Cl_{n+1, 0}, B )$ is  represented by the 
bounded Kasparov module
\[
  \Big( \Cl_{n+1, 0}, \,  B_B \otimes \bigwedge\nolimits^{\! *} \C, \, \wt{J}_1 \otimes \rho \Big)
\]
with $\Cl_{n+1,0}$-generators $\{\one\otimes \gamma, \kappa_1\otimes \rho, \ldots, \kappa_n \otimes \rho\}$.
\end{prop}
\begin{proof}
The class $\big[ \calC_{J_0}(J_1)\big]$ can equivalently be represented by the 
van Daele $K$-theory element $[J_1 \otimes \rho]- [J_0 \otimes \rho] \in DK(A \otimes \Cl_{0,n+1})$. 
Taking $J_0$ and the lift $\wt{J}_0$ as a basepoint, 
we can therefore apply Proposition \ref{prop:bdry_map_is_nice} to obtain the result.
\end{proof}

Let us also note the special situation with no additional symmetries.
\begin{prop}
The image of $\big[ \calC_{J_0}(J_1)\big] \in KO_2(A^\frakr)$ under the boundary 
map $\partial: KO_2(A^\frakr) \to KO_1(B^\frakr)$ is given by $\big[e^{\pi \wt{J}_1} \big]$.
\end{prop}
\begin{proof}
Applying~\cite[Proposition 5.7]{BKRCayley}, the boundary map can be represented by the unbounded Real Kasparov module
\[
  \Big( \C, \, \ol{ \coth\big( \tfrac{\pi}{2} \wt{J}_1 \big) B} \otimes {\Cl_{0,1}}_{\Cl_{0,1}}, \, \tanh\big( \tfrac{\pi}{2} \wt{J}_1 \big)  \otimes \rho \Big).
\]
Then the skew-adjoint Cayley transform $(T+1)(T-1)^{-1} \in \mathbb{K}_B\big( \ol{ \coth( \tfrac{\pi}{2} \wt{J}_1 ) B} \big)^\sim$ 
for $T = \tanh\big( \tfrac{\pi}{2} \wt{J}_1 \big)$
gives an isomorphism to $KO_1(B^\frakr)$ (cf. \cite[Example 5.1]{BKRCayley}), 
where we see that 
\[
   \big[ \big(\tanh( \tfrac{\pi}{2} \wt{J}_1 )+1)(\tanh( \tfrac{\pi}{2} \wt{J}_1 )-1)^{-1} \big] 
   = \big[ e^{\pi \wt{J}_1} \big] \in KO_1( B^\frakr).   \qedhere
\]
\end{proof}

Lastly we consider the $G$-equivariant case, where the same proof as Proposition \ref{prop:Boundary_KK_class} 
applies with the extra considerations done in Section \ref{subsec:G_equiv_index_general}.
\begin{prop}
Suppose that the short exact sequence of Equation \eqref{eq:Bulk_Defect_Extension} is $G$-equivariant under linear/anti-linear actions 
on $B$, $E$ and $A$ relative to $\eta:G\to \Z_2$ and such that
$\beta^E_g( \wt{J}_k) = (-1)^{\nu(g)} \wt{J}_k$, $k \in \{0,1\}$. Then 
the image of  $\big[ \calC_{J_0}(J_1)\big] \in KKR^G(\Cl_{2,0},A)$ under the boundary map 
$\partial: KKR^G( \Cl_{2,0}, A) \to KKR^G(\Cl_{1,0}, B)$ is represented by the $G$-equivariant Real Kasparov module
\[
  \Big( \Cl_{1,0}, \,  B_B \otimes \bigwedge\nolimits^{\! *} \C, \, \wt{J}_1 \otimes \rho \Big)
\]
with Clifford generator $\one \otimes \gamma$ and group action 
$\alpha_g( b \otimes v) = \beta^B_g(b) \otimes \gamma^{\nu(g)} w$ for all $b \in B$, 
$w \in \bigwedge^*\C$ and $g \in G$.
\end{prop}

While the $G$-equivariant case follows quite easily from our general boundary map computations, a more careful analysis 
is required when the extension and defect algebras $E$ and $B$ carry different linear/anti-linear group symmetries 
from $A$, as is often the case when considering systems with defects.

\section{Local equivalence, $K$-homology and coarse assembly} \label{sec:Coarse_and_KHom}

Our aim for this section is to connect the $K$-theoretic indices constructed in Section \ref{sec:KO_index_general} 
to $K$-homology in a physically meaningful way.
Due to its close connection with duality theory~\cite[Chapter 5--6]{HigsonRoe}, coarse geometry is a natural setting to 
consider this connection.

\subsection{Coarse geometry and indices}

We give a brief overview of  Roe $C^*$-algebras and the coarse index. A 
more comprehensive introduction can be found in~\cite{HigsonRoe, RoeMemoirs, Roe96}.
We consider a complex Hilbert space $\calH$ with a real involution 
$\Gamma=\Gamma^*= \Gamma^{-1}$ and mutually anti-commuting skew-adjoint 
unitaries $\{\kappa_j\}_{j=1}^n\subset \calB(\calH)$ that are
invariant under the real structure $\mathrm{Ad}_\Gamma$ on $\calB(\calH)$. 
Now suppose there is a skew-adjoint operator $F \in \calB(\calH)$ such that
\[
 \Gamma F\Gamma = F, \qquad \qquad  \one + F^2 \in \mathbb{K}(\calH), \qquad \qquad  F\kappa_j = -\kappa_j F.
\]
Then $F$ is Fredholm and we construct a Real Fredholm module and $K$-homology class
\[
 [F] = \Big[ \big( \Cl_{n+1,0}, \, \calH \otimes \bigwedge\nolimits^{\! *} \C, \, F \otimes \rho \big) \Big] \in KO^{-n-1}(\R)
\]
with Clifford generators $\big\{ \one \otimes \gamma, \kappa_1 \otimes \rho, \ldots, \kappa_n \otimes \rho \big\}$. 
We can also  define an ungraded Clifford module index $[\Ker(F)] \in \calM_{n}/ \calM_{n+1} \cong KO_{n+1}(\R)$. 
If $[F]=[F'] \in KO^{-n-1}(\R)$, then $[\Ker(F)] = [\Ker(F')] \in KO_{n+1}(\R)$. Therefore the 
skew-adjoint Fredholm index can be seen as a map $KO^{-n-1}(\R) \to KO_{n+1}(\R)$. 
The coarse assembly map generalises this basic idea.

Let $X$ be a second countable, metrizable and locally compact space with proper coarse structure. 
For the reader unfamiliar with coarse structures, it will suffice to consider $X$ as a second countable, locally compact and 
proper metric space (closed and bounded subsets of $X$ are compact). 
We also assume that $C_0(X)$ has a real structure $\frakr$, e.g. $f(x)^\frakr = \ol{f(\tau(x))}$ with 
$\tau$ an order-$2$ automorphism on $X$.
Fix a Hilbert space $\calH$ and real involution $\Gamma$ and suppose there is a 
non-degenerate Real representation $\varphi: C_0(X) \to \calB(\calH)$, 
$\varphi(f^\frakr) = \mathrm{Ad}_\Gamma \circ \varphi(f)$. We say that 
$\varphi$ is ample if no non-zero element of $C_0(X)$ acts compactly on $\calH$. 

\begin{example}
Let $M$ be a complete Riemannian manifold and $V$ a complex Hermitian vector bundle. 
Suppose that $C_0(M)$ has a real structure $f(x)^\frakr = \ol{f( \tau(x))}$ with $\tau$ an order-$2$ automorphism 
on $M$. Then taking $\mathfrak{C}$ pointwise complex conjugation on $V$, we can define a 
real structure on the space of $L^2$-sections $L^2(M, V)$ by $\psi(x)^\frakr = \mathfrak{C} \circ \psi( \tau(x))$ for 
all $\psi \in L^2(M,V)$. 
The multiplication representation of $C_0(M)$ on $L^2(M, V)$ is then a Real and ample representation.
To see this, we remark that the spectrum of $\phi(f)$ is the closure of $\Ran(f)$, which is a 
non-discrete set for $f$ non-zero and $M$ complete. 
\end{example}

\begin{defn}
Let $T \in \calB(\calH)$ and $\varphi: C_0(X) \to \calB(\calH)$ a representation.
\begin{enumerate}
  \item[(i)] We say that $T$ is pseudolocal with respect to $\varphi: C_0(X) \to \calB(\calH)$  
  if $\varphi(f_1) T \varphi(f_2)$ is compact for all $f_1, f_2 \in C_0(X)$ 
  such that  $f_1$ or $f_2$ have compact support and $\mathrm{supp}(f_1) \cap \mathrm{supp}(f_2) = \emptyset$.
  \item[(ii)] We say that $T$ has finite propagation (or $T$ is controlled) with respect to $\varphi: C_0(X) \to \calB(\calH)$ 
   if there exists an $R>0$ such that 
  $\varphi(f_1) T \varphi(f_2) = 0$ for any $f_1, f_2 \in C_0(X)$ such that $d(\mathrm{supp}(f_1), \mathrm{supp}(f_2) ) > R$.
\end{enumerate}
The $C^*$-subalgebra of $\calB(\calH)$ generated operators that are pseudolocal and have 
finite propagation with respect to $\varphi: C_0(X) \to \calB(\calH)$  is denoted by $D^*(X)$.
\end{defn}

The real involution $\Gamma$ on $\calH$ gives a real structure $\mathrm{Ad}_\Gamma$ on $D^*(X)$ that 
makes it a Real $C^*$-algebra.

\begin{lemma}[{Kasparov's Lemma, \cite[Lemma 5.4.7]{HigsonRoe}}]
A bounded operator $T$ on $\calH$ is pseudolocal with respect to $\varphi: C_0(X) \to \calB(\calH)$ 
if and only if $[T, \varphi(f)] \in \mathbb{K}(\calH)$ for all $f\in C_0(X)$.
\end{lemma}

\begin{defn}
We say that $T \in \calB(\calH)$ is locally compact with respect to $\varphi: C_0(X) \to \calB(\calH)$ 
 if for all $f \in C_0(X)$, $\varphi(f)T$ and $T \varphi(f)$ are 
compact. We define the Roe algebra $C^*(X) = \ol{ \{T \in D^*(X)\,:\, T \text{ locally compact} \}}$.
\end{defn}

It follows from the definition that $C^*(X)$ is a closed two-sided ideal in $D^*(X)$. Note also that the definition of 
$D^*(X)$ and $C^*(X)$ rely on a choice of Hilbert space $\calH$ and representation $\varphi$. 
If $\varphi_1$ and $\varphi_2$ are ample representations, then 
$C^*(X, \varphi_1)\cong C^*(X,\varphi_2)$~\cite[Theorem 1]{EwertMeyer}. 

If $X$ is discrete and we take $\varphi: C_0(X) \to \ell^2(X)\otimes \C^n$ as 
$\varphi(f)(\psi(x)\otimes v) = f(x)\psi(x) \otimes v$ for $\psi \otimes v \in \ell^2(X) \otimes \C^n$, 
then  this representation is \emph{not} ample. 
In this setting, a natural object of study is the \emph{uniform} Roe algebra $C^*_u(X)$, which embeds 
in $C^*(X)$.

Because we consider both the discrete and continuous settings, we will work primarily with $C^*(X)$ with 
the knowledge that $C^*_u(X)$ can be embedded in this algebra if $X$ is discrete. At the 
level of $K$-theory, $KO_\ast(C^*_u(X)^\frakr)$ is much richer than $KO_\ast(C^*(X)^\frakr)$, though 
the $K$-theory of $C^*(X)$ will still capture the large scale  properties of $X$. 
See~\cite{Kubota17,EwertMeyer} for further discussion on this point.

Let us now consider the coarse index, which is a map  $KO^{-\ast}(C_0(X)^\frakr) \to KO_{\ast}(C^*(X)^\frakr)$. 
We first  note the following.

\begin{prop} \label{prop:anti-commuting_F_to_Khom}
Suppose that $F=F^\frakr = -F^* \in \calB(\calH)$ and 
assume one of the following:
\begin{enumerate}
  \item[(a)] There is an ample Real representation $\varphi: C_0(X) \to \calB(\calH)$ and $\{\kappa_j\}_{j=1}^n \subset \calB(\calH)$ are 
  mutually anti-commuting Real skew-adjoint unitaries with infinite-dimensional eigenspaces such that 
  for all $f\in C_0(X)$ and $j= 1,\ldots, n$, 
  \[
  [F, \varphi(f)], \,\, \varphi(f)(\one+ F^2), \,\,  [\kappa_j, \varphi(f)] \in \mathbb{K}(\calH), \qquad \qquad F \kappa_j = - \kappa_j F ,
\]
  \item[(b)] There is a Real representation $\varphi: C_0(X) \to \calB(\calH)$ and $\{\kappa_j\}_{j=1}^n \subset \calB(\calH)$ are 
  mutually anti-commuting Real skew-adjoint unitaries such for all $f\in C_0(X)$ and $j= 1,\ldots, n$, 
  \[
  [F, \varphi(f)], \,\, \varphi(f)(\one+ F^2)  \in \mathbb{K}(\calH), \qquad   [\kappa_j, \varphi(f)] = 0, \qquad F \kappa_j = - \kappa_j F .
\]
\end{enumerate}
Then the triple
\[
   \Big( C_0(X) \otimes \Cl_{n+1,0}, \, {}_\varphi \calH \otimes  \bigwedge\nolimits^{\! *} \C, \, F \otimes \rho \Big)
\]
is a Real Fredholm module (Real $C_0(X)$--$\C$ Kasparov module) with Clifford generators 
$\{\one \otimes \gamma, \kappa_1\otimes \rho, \ldots, \kappa_n \otimes \rho\}$. 
\end{prop}
\begin{proof}
Given condition (b), the result holds immediately from the definition of a Real Fredholm module. For condition (a), we need to 
check that the specified representation of $C_0(X) \otimes \Cl_{n+1,0}$ is well-defined. Because $\varphi$ is ample, $\varphi(f) = q\circ \varphi(f)$ 
for $q: \calB(\calH) \to \calB(\calH)/\mathbb{K}(\calH)$ the quotient onto the Calkin algebra. 
For $j \in \{1,\ldots, n\}$,  the skew-adjoint unitary $\kappa_j$ has infinite-dimensional 
eigenspaces, which implies that $q (\kappa_j) = \kappa_j$. We therefore have that 
\[
  [\varphi(f), \kappa_j] = [ q\circ \varphi( f), q(\kappa_j) ] = q\big( [\varphi(f), \kappa_j] \big) = 0, \quad \text{for all } j\in \{1,\ldots, n\}.
\]
Therefore the representation of $C_0(X) \otimes \Cl_{n+1,0}$ is indeed well-defined. The remaining conditions required to obtain 
a Real Fredholm module are immediate from the assumptions.
\end{proof}

Given the setting of Proposition \ref{prop:anti-commuting_F_to_Khom}, we let $[F]$ denote the corresponding 
class in $KKR(C_0(X) \otimes \Cl_{n+1,0}, \C) \cong KO^{-n-1}(C_0(X)^\frakr)$. 
When $f(x)^\frakr = \ol{f(x)}$, we have that $KO^{-n}(C_0(X)^\frakr) \cong KO_n(X)$,
the topological $K$-homology of the space $X$.

Let us now define the coarse index or coarse assembly map. Assume the setting of Proposition \ref{prop:anti-commuting_F_to_Khom}.
If $F \in \calB(\calH)$ does not have finite propagation  with respect to $\varphi:C_0(X)\to \calB(\calH)$, 
take a Real partition of unity $\{\eta_i\}$  of $X$ subordinate to a locally finite open cover and define 
\[
   F' = \sum_i \eta_i^{1/2} F \eta_i^{1/2},
\]
which converges in the strong topology. The operator $F'$ also determines a Real Fredholm module and 
$[F] = [F'] \in KO^{-n-1}(C_0(X))$~\cite[Page 354]{HigsonRoe}. Therefore $F'$ is an element of $D^*(X)$ such 
that $(F')^2+\one \in C^*(X)$. Hence $q(F')$ is a skew-adjoint unitary in the quotient $D^*(X)/C^*(X)$. 
In particular, $q(F') \otimes \rho \in D^*(X)/C^*(X) \otimes \Cl_{0,1}$ is an odd self-adjoint unitary that anti-commutes 
 with the odd self-adjoint unitaries $\{q(\kappa_j)\otimes \rho\}_{j=1}^n$. Recalling Remark \ref{rk:extra_clifford_gens_shifts_the_degree} 
 and fixing a basepoint odd Real self-adjoint unitary $e \in D^*(X)/C^*(X) \otimes \Cl_{0,1}$ that anti-commutes with 
$\{q(\kappa_j)\otimes \rho\}_{j=1}^n$, we therefore obtain a 
class in the van Daele $K$-theory group $[q(F')\otimes \rho]\in DK_e( D^*(X)/C^*(X) \otimes \Cl_{0,n+1} )$.
Following Roe~\cite{Roe02,Roe04}, the coarse index $\mu_X: KO^{-n-1}(C_0(X)^\frakr) \to KO_{n+1}(C^*(X)^\frakr)$ can be 
defined by the composition 
\[
  KO^{-n-1}( C_0(X)^\frakr)  \to DK_e( D^*(X)/C^*(X) \otimes \Cl_{0,n+1} ) \xrightarrow{\calC \circ \partial}   KKR( \Cl_{n+1,0}, C^*(X)),
\]
where we have identified $KKR( \Cl_{n+1,0}, C^*(X))$ with $KO_{n+1}(C^*(X)^\frakr)$ and 
$\calC \circ \partial$ is the boundary map in van Daele $K$-theory composed with the 
equivalence between $DK$ and $KKR$.
This map can also be defined in the equivariant setting 
(see Sections \ref{sec:Compact_G_assembly} and \ref{sec:cocompact_assembly} below).

\begin{example}
Suppose that $X$ is a second countable, metrizable and \emph{compact} space with proper coarse structure.  
Then $C^*(X) \cong \mathbb{K}(\calH)$ 
and the assembly map $\mu_X: KO^{-n-1}(C(X)^\frakr)\to KO_{n+1}(C^*(X)^\frakr)$ reduces 
to the Clifford module index map $KO^{-n-1}(\R) \ni [F] \mapsto [\Ker(F)] \in KO_{n+1}(\R)$ considered at the beginning of this section.
\end{example}

\subsection{Locally equivalent ground states}

Let us now fix a Real non-degenerate representation $\varphi:C_0(X) \to \calB(\calH)$.
We are interested in ground states on $A^\mathrm{car}_\mathrm{sd}(\calH,\Gamma)$ with respect 
to the quasifree action $\alpha: \R \to \Aut[ A^\mathrm{car}_\mathrm{sd}(\calH,\Gamma)]$ such that 
$\alpha_t( \frakc(v) ) = \frakc\big( e^{itH} v\big)$ for $t \in \R$, $v \in \calH$ and $H =H^*=-\Gamma H \Gamma$ 
 a BdG Hamiltonian. We will now
 restrict to BdG Hamiltonians and dynamics that are local with respect to the representation $\varphi$.

\begin{assumption} \label{assump:SingleHamilt}
We consider  BdG Hamiltonians $H=H^*=-\Gamma H \Gamma$ such that $0\notin \sigma(H)$.
\begin{enumerate}
  \item[(a)] (Unbounded case) If $H$ is unbounded, we assume that $\chi(H)  \in \mathrm{Mult}(C^*(X))$ 
  for any normalising function $\chi$ and $[\varphi(f), i\, \mathrm{sgn}(H)] \in \mathbb{K}(\calH)$ for all $f \in C_0(X)$.
  \item[(b)] (Bounded case) If $H$ is bounded, we assume that $H$ is invertible in $D^*(X)$.
\end{enumerate}
\end{assumption}

 If $\chi(H) \in D^*(X)$ for any regularising function $\chi$, then Assumption \ref{assump:SingleHamilt}(a) 
   is satisfied. Our slightly weaker condition allows us to accommodate Hamiltonians 
   with higher-order terms, see Example \ref{ex:unbdd_examples}.

The functional calculus gives us the following.
\begin{lemma}
Let $H \in \calB(\calH)$ be gapped and pseudolocal with respect to $\varphi:C_0(X)\to \calB(\calH)$. 
Then for any regularising function $\chi$,  $\chi(H)$, $J=iH|H|^{-1}$ and $P = \chi_{[0,\infty)}(H)$ 
are pseudolocal with respect to $\varphi:C_0(X)\to \calB(\calH)$.
\end{lemma}

\begin{example}[Bounded/discrete examples] \label{ex:bdd_examples}
Let $\Lambda$ be a proper and discrete metric space, e.g. $\Lambda=\Z^d$. In this case, $C_0(\Lambda)$ acts naturally on 
$\calV=\ell^2(\Lambda)\otimes \C^{n}$, which we can extend to $\calH = \calV \otimes \C^{2}$ 
with real involution $\mathfrak{C}(\one \otimes \sigma_1)$. We consider BdG Hamiltonians as in 
Equation \eqref{eq:Generic_BdG},
\[
  H= \begin{pmatrix} h & \delta \\ \delta^* & - \ol{h} \end{pmatrix}, \qquad h, \, \delta \in \calB[ \ell^2(\Lambda)\otimes \C^n], 
   \qquad \ol{A} = \mathfrak{C} A \mathfrak{C}, \qquad \delta^* = - \ol{\delta}.
\]
The representation $\varphi: C_0(\Lambda)\to \calB[ \ell^2(\Lambda)\otimes \C^n ]$ can be decomposed into a 
sum of terms of $\{p_x\}_{x\in \Lambda}$, the projection onto the site $\{x\} \otimes \C^n$.
With this in mind, we impose the following conditions for $T= h$ or $\delta$:
\begin{itemize}
  \item There exists $R>0$ such that $p_y T p_x = 0$ if $d(x,y)>R$,
  \item The operators $p_x T$ and $T p_x$ are compact for any $x \in \Lambda$.
\end{itemize}
If these conditions are satisfied for $h$ and $\delta$, then $H$ is an element of the \emph{uniform} Roe algebra, 
$C^*_u(\Lambda)$, which can be naturally embedded in $C^*(\Lambda)$. In particular, if $H$ is invertible in $C^*_u(\Lambda)$, 
then it easily follows that $\chi(H)$ and $iH|H|^{-1}$ are elements in $C^*_u(\Lambda)\subset D^*(\Lambda)$ for 
any normalising function $\chi$. 
As such, for two BdG Hamiltonians $H_0, \, H_1$ satisfying these 
conditions, we obtain that $H_0-H_1, \, J_0 - J_1 \in C^*(\Lambda)$.
\end{example}

\begin{example}[Unbounded/continuous examples] \label{ex:unbdd_examples}
Let $M$ be a complete Riemannian manifold  and consider the multiplication representation of $C_0(M)$ 
acting diagonally on $\calH = L^2(M, \C^n) \otimes \C^2$ with 
real involution $\mathfrak{C}(\one \otimes \sigma_1)$ and $\mathfrak{C}$ complex conjugation. 
As before, we restrict to gapped BdG Hamiltonians 
of the form $H = \begin{pmatrix} h & \delta \\ \delta^* & -\ol{h} \end{pmatrix}$ with 
$\delta^* = - \ol{\delta}$. In typical examples, $h$ and $\delta$ are differential 
operators (cf. Example \ref{ex:canonical_example}). 

To apply our general framework, we would like to find gapped  BdG Hamiltonians such that 
$J= i\,\mathrm{sgn}(H)$ is pseudolocal and an element of $\mathrm{Mult}(C^*(M))$. We 
also wish to consider pairs of BdG Hamiltonians such that $J_0-J_1 \in C^*(M)$. We will 
examine these conditions separately.

\vspace{0.1cm}

To show $J= i\,\mathrm{sgn}(H)$ is pseudolocal, 
some care needs to be taken if $H$ is a second or higher-order differential operator as the commutator 
$[H, \varphi(f)]$ will generally not be bounded for $f \in C_c^\infty(M)$. 
Such BdG Hamiltonians will not give rise to 
a spectral triple (unbounded $C_0(M)$--$\C$ Kasparov module), 
whose properties can then 
be used to show that $\mathrm{sgn}(H)$ is pseudolocal. 
However, following~\cite[Appendix A]{GM14} and fixing $\varepsilon >0$, 
we ask for BdG Hamiltonians to have 
$\varepsilon$-bounded commutators in the sense that for all $f \in C_c^\infty(M)$,
\begin{enumerate}
  \item[(i)] $\varphi(f) \Dom(H) \subset \Dom(H)$,
  \item[(ii)] The operators $[H, \varphi(f)](1+H^2)^{-\frac{1-\varepsilon}{2}}$ and $(1+H^2)^{-\frac{1-\varepsilon}{2}}[H, \varphi(f)]$ 
  extend to bounded operators on $\calH$.
\end{enumerate}
The basic intuition behind the parameter $\varepsilon$ is as follows. Suppose that $H$ is an order-$m$ differential 
operator and take $f \in C_c^\infty(M)$. Then the commutator $[H, \varphi(f)]$ will have order $m-1$. Therefore 
$[H, \varphi(f)](1+H^2)^{-\frac{m-1}{2m}}$ and $(1+H^2)^{-\frac{m-1}{2m}}[H, \varphi(f)]$ are order-$0$ 
pseudodifferential operators and so are bounded on $L^2(M, \C^{2n})$. Hence $H$ is $\varepsilon$-bounded 
with $\varepsilon = \tfrac{1}{m}$. More generally, 
the parameter $\varepsilon^{-1}$ can be thought of as the order of the operator $H$.

For BdG Hamiltonians with $\varepsilon$-bounded commutators and whose resolvent is locally compact with respect to $\varphi$, the 
bounded transform $F_H = H( 1+H^2)^{-1/2}$ is such that the triple 
\[
  \Big( C_0(M) \otimes \Cl_{1,0}, \, \calH \otimes \bigwedge\nolimits^{\! \ast} \C, \, i F_H \otimes \rho \Big)
\] 
will give a Real Fredholm module~\cite[Theorem A.6]{GM14}. Furthermore, because $H$ is invertible, 
$[0,1] \ni t \mapsto F_t = H\big( (1-t) + H^2 \big)^{-1/2}$ is a norm-continuous path between 
$F_H$ and $\mathrm{sgn}(H)$. Therefore for any $f\in C_0(M)$, $[\varphi(f), \mathrm{sgn}(H)]$ is a norm-limit of compact 
operators and so is also compact. Hence $J = i \,\mathrm{sgn}(H)$ is pseudolocal with respect to $\varphi$.

\vspace{0.1cm}

Let us now consider the condition that $J \in \mathrm{Mult}(C^*(M))$ for BdG Hamiltonians 
$H = \begin{pmatrix} h & \delta \\ \delta^* & -\ol{h} \end{pmatrix}$.
If $\delta$ a first-order differential operator and 
$h = \Delta_M - \mu$ with $\Delta_M$ the Laplace--Beltrami operator, it is in general quite 
non-trivial to show that $J \in \mathrm{Mult}(C^*(M))$.
In the case that $M = \R^d$, the operator $\Delta_{\R^d}+V$ is affiliated to the Roe algebra 
$C^*(\R^d)$ for any $V \in L^\infty(\R^d)$~\cite[Proposition 1]{EwertMeyer}. The  
first-order coupling term $\delta$ will also be 
affiliated to $C^*(\R^d)$. Therefore, provided there is sufficient compatibility with $h$ and $\delta$ (e.g. 
$h$ and $\delta$ commute) and 
$H$ is gapped, $H^{-1}$ will then be a matrix over 
$\mathrm{Mult}(C^*(\R^d))$. Ignoring matrix degrees of freedom, we then obtain that  
 $J \in \mathrm{Mult}(C^*(\R^d))$.
For the case of more general $M$ we will simply assume that $J \in \mathrm{Mult}(C^*(M))$.

\vspace{0.1cm}

Given a pair of BdG Hamiltonians $H_0$ and $H_1$ such that $J_0$ and $J_1$ 
are pseudolocal elements in $\mathrm{Mult}(C^*(M))$, 
we furthermore  assume that $\Dom(H_0)=\Dom(H_1)$ and 
$\varphi(f)(H_0-H_1)(i+ H_0)^{-1} \in \mathbb{K}(\calH)$ 
for all $f \in C_0(M)$. Then  $J_0-J_1 \in C^*(M)$ and 
we are in the setting of Section \ref{sec:KO_index_general}.
In more general examples, the framework 
of~\cite{CGPRS} provides technical conditions on $H_0$ and $H_1$ so that 
$J_0-J_1 \in C^*(M)$.
\end{example}

\begin{defn} \label{def:local_equiv_state}
\begin{enumerate}
  \item[(i)] We say that the representation $\varphi: C_0(X) \to \calB(\calH)$ is locally compatible with $\Gamma$ if 
  $\Gamma$ restricts to a real involution on $\ol{\Ran(\varphi(f))}\subset \calH$ for all $f\in C_0(X)$.
  \item[(ii)] We say that two pure and quasifree states $\omega_0$ and $\omega_1$ of $A^\mathrm{car}_\mathrm{sd}(\calH,\Gamma)$ 
  are locally equivalent with respect to $(X, \varphi)$ if $\varphi$ is locally compatible with $\Gamma$ 
and there is a dense $\ast$-algebra $\calA \subset C_0(X)$ such that 
$\omega_0$ is equivalent to $\omega_1$ 
as a state on $A^\mathrm{car}_\mathrm{sd}\big( \ol{\Ran(\varphi(f))}, \Gamma \big)$ for all $f \in \calA$.
\end{enumerate}
\end{defn}

A physically reasonable choice for a dense $\ast$-subalgebra $\calA \subset C_0(X)$ is $C_c(X)$, the algebra 
of compactly supported functions. 
In our examples of interest, $C_0(X)$ acts diagonally on 
$\calH = \calV \otimes \C^2$ with $\Gamma = \mathfrak{C}(\one \otimes \sigma_1)$. 
Letting $\Pi_{\varphi(f)}$ denote the projection onto $\ol{\Ran(\varphi(f))}$, 
$\Pi_{\varphi(f)} \Gamma = \Gamma \Pi_{\varphi(f)}$ and $\varphi$ is locally compatible with $\Gamma$ 
when $\mathfrak{C}$ is pointwise complex conjugation for example.
Recalling Theorem \ref{thm:Araki_quasifree}, the basis projections $P_0$ and $P_1$ 
will give locally equivalent quasifree states  $\omega_0$ and $\omega_1$ if and only 
if $\Pi_{\varphi(f)} (P_0 - P_1) \Pi_{\varphi(f)}$ is Hilbert-Schmidt for all $f\in \calA$.

\begin{remarks}
\begin{enumerate}
  \item[(i)] If $\omega_0$ is unitarily equivalent to $\omega_1$, then it is locally equivalent for any $(X,\varphi)$ as in 
  Definition \ref{def:local_equiv_state}.
  \item[(ii)] If the space $X$ is compact, then our notion of local equivalence reduces to equivalence of the states 
$\omega_0$ and $\omega_1$.
  \item[(iii)] If the representation $\varphi: C_0(X) \to \calB(\calH)$ \emph{commutes} with the BdG Hamiltonians, 
  $\varphi(f)H_k = H_k\varphi(f)$ for $k \in \{0,1\}$, then $\omega_0$ and $\omega_1$ are locally equivalent if 
  and only if they are equivalent.
  \item[(iv)] A similar notion of local quasiequivalence is defined for states of nets of $C^*$-algebras 
  $\calO \to \mathfrak{A}(\calO)$ that appear in algebraic quantum field theory, see~\cite{LudersRoberts, DAntoniHollands} for example. 
  Here we work with a different class of spaces.
\end{enumerate}
\end{remarks}

\begin{example}
Let us revisit the case of a discrete and proper metric space $\Lambda$ from Example \ref{ex:bdd_examples}. 
Take two invertible BdG Hamiltonians 
$H_0, \, H_1 \in C^*_u(\Lambda) \subset C^*(\Lambda)$ acting on $\ell^2(\Lambda, \C^n)\otimes \C^{2}$ 
and $\Gamma = \mathfrak{C}(\one \otimes \sigma_1)$ with $\mathfrak{C}$ component-wise complex conjugation. 
We can consider $C_c(\Lambda)$ as a dense 
$\ast$-subalgebra of $C_0(\Lambda)$. In particular, any function $f \in C_c(\Lambda)$ will be supported 
on a {\em finite} set $Y \subset \Lambda$. Hence, the restriction of $H_0$ and $H_1$ to 
$\ol{\Ran(\varphi(f))}$ is the restriction to $\ell^2( Y) \otimes \C^{2n} \cong \C^{|Y|} \otimes \C^{2n}$. 
Because we are now in a finite-dimensional Hilbert space, all pure  states are unitarily equivalent to each other. 
Therefore we see that in discrete examples, local equivalence is satisfied without issue.
\end{example}

\begin{prop} \label{prop:local_equiv_to_Roe_alg}
If the gapped ground states $\omega_0$ and $\omega_1$ are locally equivalent 
with respect to $(X, \varphi)$, then $\varphi(f)(J_0 - J_1) \in \mathbb{K}(\calH)$  for all $f \in C_0(X)$, i.e., $J_0-J_1 \in C^*(X)$. 
\end{prop}
\begin{proof}
Choose $f\in C_0(X)$ with an approximating sequence $f_n \in \calA$. Without loss of generality, we can assume that 
$f$ is  real-valued.
As $\omega_0$ and $\omega_1$ are locally equivalent with respect to $(X,\varphi)$ 
$\varphi(f_n)(P_0 -P_1)\varphi(f_n)$ maps the unit ball of $\calH$ to a precompact set. Hence it is compact 
and so is $\varphi(f_n)(J_0-J_1)\varphi(f_n)= 2i \varphi(f_n)(P_0-P_1)\varphi(f_n)$.

We will first show that $\big| \varphi(f)(J_0-J_1) \big| \in \mathbb{K}$, where we compute 
\begin{align*}
  \big( \varphi(f)(J_0- J_1) \big)^* \varphi(f)(J_0-J_1) &= -(J_0-J_1)\varphi(f)^2 (J_0-J_1) \\
  &= -\varphi(f)(J_0-J_1)\varphi(f)(J_0-J_1)  + \mathbb{K} \\
  &= \lim_{n\to \infty} - \varphi(f_n) (J_0 - J_1) \varphi(f_n) (J_0-J_1) + \mathbb{K},
\end{align*}
where in the second line we used that $[\varphi(f),(J_0-J_1)] \in \mathbb{K}$ as $J_0-J_1$ is pseudolocal 
with respect to $\varphi:C_0(X)\to \calB(\calH)$. 
Hence we have a limit of compact operators which is compact. Because the compact operators are 
closed under square root, this implies that $\big| \varphi(f)(J_0-J_1) \big| \in \mathbb{K}$. By the polar 
decomposition, $\varphi(f)(J_0-J_1)$ is compact. 
\end{proof}

\subsection{Local equivalence to $K$-homology class}

Building from our work in the previous subsection, we will extend Assumption \ref{assump:SingleHamilt} 
and  consider \emph{pairs} of BdG Hamiltonians that are local with respect to the Real representation 
$\varphi:C_0(X)\to \calB(\calH)$.

\begin{assumption} \label{assump:ControlledHamiltonians}
We assume that the BdG Hamiltonians $H_0, \, H_1$ are such that for all $k \in \{0,1\}$,
\begin{enumerate}
  \item[(i)] $0 \notin \sigma(H_k)$  and $H_k = H_k^* = -\Gamma H_k \Gamma$,
  \item[(ii)] $\chi(H_k)  \in \mathrm{Mult}(C^*(X))$ for any normalising function $\chi$ and 
  $J_k  = i\,\mathrm{sgn}(H_k)$ is pseudolocal with respect to $\varphi:C_0(X) \to \calB(\calH)$,
  \item[(iii)]  $\varphi(f)(J_0-J_1) \in \mathbb{K}(\calH)$ for 
  all $f \in C_0(X)$; i.e. $J_0-J_1 \in C^*(X)$.
\end{enumerate}
If $H_0$ and $H_1$ are bounded, we assume they are invertible as elements of $D^*(X)$. 
\end{assumption}

\begin{remarks}
\begin{enumerate}
  \item[(i)] By Proposition \ref{prop:local_equiv_to_Roe_alg}, pairs of invertible BdG Hamiltonians $H_0, \, H_1$ 
  satisfying Assumption \ref{assump:SingleHamilt} and whose ground states are locally equivalent 
   will satisfy the conditions of Assumption \ref{assump:ControlledHamiltonians}. 
  \item[(ii)] If $X$ is a compact space, then $C^*(X)= \mathbb{K}(\calH)$ and we recover the setting of Section \ref{sec:Basic_Quasifree_index}.
\end{enumerate}
\end{remarks}

We now use Assumption \ref{assump:ControlledHamiltonians} and the relative Cayley transform $\calC_{J_0}(J_1)$ 
to construct a $K$-homology class from the pair of gapped ground 
states $\omega_0$ and $\omega_1$.

\begin{prop} \label{prop:Coarse_KHom_class_noSymm}
Suppose that $H_0$ and $H_1$ satisfy Assumption \ref{assump:ControlledHamiltonians}.
Further assume one of the following:
\begin{enumerate}
  \item[(a)] The Real representation $\varphi:C_0(X) \to \calB(\calH)$ is ample,
  \item[(b)] The Hamiltonian $H_0$ is trivially local with respect to $(X,\varphi)$ in the sense that 
  $[\varphi(f), J_0] = 0$ for all $f \in C_0(X)$. 
\end{enumerate}
Then the triple
\[
  \Big( C_0(X) \otimes \Cl_{2,0}, \, {}_\varphi \calH \otimes \bigwedge\nolimits^{\! *} \C, \, F_{\calC_{J_0}(J_1)} \otimes \rho \Big)
\]
is a Real Fredholm module with Clifford generators $\{\one \otimes \gamma, J_0 \otimes \rho\}$ and 
$F_{\calC_{J_0}(J_1)}$ the bounded transform of the Cayley transform $\calC_{J_0}(J_1) = J_0(J_1 + J_0)(J_1-J_0)^{-1}$.
\end{prop}
\begin{proof}
We first consider case (a). Because $J_0$ has infinitely degenerate $\pm i$ eigenspaces, it generates an ample 
representation of the ungraded Clifford algebra $\Cl_{0,1}$ on $\calH$. Therefore, denoting 
$q: \calB(\calH) \to \calB(\calH) / \mathbb{K}(\calH)$ the quotient map, 
$[\varphi(f), J_0] = q\big([ \varphi(f), J_0] \big) = 0$ for all $f\in C_0(X)$ 
 as $\varphi$ is ample and $[ \varphi(f), J_0]$ is compact. Hence the representation of $C_0(X)\otimes \Cl_{2,0}$ is well-defined.
 
Recall the bounded transform $F_{\calC_{J_0}(J_1)} = \tfrac{1}{2} J_0(J_1J_0-1)(J_0J_1+1)^{-1}(2+J_0J_1+J_1J_0)^{1/2}$ and 
its basic properties from Equations \eqref{eq:Skew_Cayley_F} and \eqref{eq:Skew_Cayley_F_properties} on Page \ref{eq:Skew_Cayley_F_properties}. 
Because $J_1J_0$ is normal, we can write $F_{\calC_{J_0}(J_1)} =  J_0\, \eta( J_1J_0, (J_1J_0)^*)$ with 
$\eta(z, \ol{z}) = \tfrac{1}{2} (z-1)(\ol{z}+1)^{-1}(2+z+\ol{z})^{1/2}$ a continuous function on the relevant domain. 
Because $J_0$ and $J_1$ have compact 
commutators with $\varphi(f)$, $[\varphi(f), J_0 \, \eta(J_1 J_0) ]$ will also be compact for 
any $f\in C_0(X)$. We similarly have from Equation \eqref{eq:Skew_Cayley_F_properties} that
\[
  \one + F_{\calC_{J_0}(J_1)}^2  = -\frac{1}{4}(J_0-J_1)^2 \in C^*(X).  
\]
Because $J_0$ anti-commutes with $\calC_{J_0}(J_1)$, it will anti-commute with $F_{\calC_{J_0}(J_1)}$ and the conditions 
to obtain a Real Fredholm module are satisfied.

For case (b), we immediately obtain a well-defined representation of $C_0(X)\otimes \Cl_{2,0}$. The 
rest of the proof then follows the same argument as case (a).
\end{proof}

\begin{remark}
Following the perspective of SPT phases, the condition $[\varphi(f), J_0] = 0$ for all $f \in C_0(X)$ from 
case (b) of Proposition \ref{prop:Coarse_KHom_class_noSymm} specifies a trivial locally gapped system 
for which we then consider $H_1$ such that $J_0 - J_1 \in C^*(X)$. Because case (b) of Proposition \ref{prop:Coarse_KHom_class_noSymm} 
does not require an ample representation, the result can be applied for $C_0(\Lambda)$ acting on $\ell^2(\Lambda)\otimes \C^{2n}$ with 
$\Lambda$ a discrete and proper metric space.
\end{remark}

\subsection{$K$-theory classes and coarse assembly}

Given the Real representation $\varphi:C_0(X)\to \calB(\calH)$ and BdG Hamiltonians $H_0$ and $H_1$ satisfying 
Assumption \ref{assump:ControlledHamiltonians},  we assume that we are in one of the following settings:
\begin{enumerate}
  \item[(a)] $\varphi$ is ample,
  \item[(b)] $[\varphi(f),J_0]=0$ for all $f\in C_0(X)$.
\end{enumerate}
In either setting, we obtain a $K$-homology element $\big[F_{\calC_{J_0}(J_1)}^\mathrm{Hom}\big] \in KO^{-2}(C_0(X)^\frakr)$ 
by Proposition \ref{prop:Coarse_KHom_class_noSymm}.
Because $J_0, \, J_1 \in  \mathrm{Mult}(C^*(X))$ 
with $J_0-J_1 \in C^*(X)$, we can also consider the $K$-theory elements constructed in Section \ref{Subsec:General_KTheory_class_noSymm} 
for $A= C^*(X)$. Namely, by Proposition \ref{prop:skew_cayley_kk}, we have the class
\begin{align} \label{eq:Coarse_Cayley_class}
\big[\calC_{J_0}(J_1) \big] &=  \Big[ \big( \Cl_{2,0}, \, \ol{(J_0-J_1)C^*(X)}_{C^*(X)} \otimes \bigwedge\nolimits^{\!*} \C, \, \calC_{J_0}(J_1) \otimes \rho \big) \Big] 
\in KO_2(C^*(X)^\frakr),
\end{align}
which by Proposition \ref{prop:vD_class_to_Cayley_class} is equivalent to the relative van Daele $K$-theory element
\begin{equation} \label{eq:vD_class}
  \big[ J_1 \otimes \rho\big] - \big[ J_0 \otimes \rho \big] \in DK(\mathrm{Mult}(C^*(X)) \otimes \Cl_{0,1},\, \mathrm{Mult}(C^*(X))/C^*(X)\otimes \Cl_{0,1}) ,
\end{equation}
where we recall that $DK(A, A/I) \cong DK(I)$.

Our task is to relate the $K$-homology element $\big[F_{\calC_{J_0}(J_1)}^\mathrm{Hom}\big]$  to the 
 $K$-theory classes in Equations \eqref{eq:Coarse_Cayley_class} and \eqref{eq:vD_class}. To do this, 
we will use the coarse assembly map $\mu_X: KO^{-\ast}(C_0(X)^\frakr) \to KO_\ast(C^*(X)^\frakr)$.

\begin{thm} \label{thm:Coarse_Assembly_Works}
The coarse assembly map $\mu_X: KO^{-2}(C_0(X)^\frakr) \to KO_2(C^*(X)^\frakr)$ is such that 
$\mu_X\big( [ F_{\calC_{J_0}(J_1)}^\mathrm{Hom} ] \big) = \big[ \calC_{J_0}(J_1) \big]$ with 
$\big[F_{\calC_{J_0}(J_1)}^\mathrm{Hom}\big]$ the $K$-homology class from Proposition \ref{prop:Coarse_KHom_class_noSymm} and 
$\big[ \calC_{J_0}(J_1) \big]$ the $K$-theory class from Equation \eqref{eq:Coarse_Cayley_class}.
\end{thm}
\begin{proof}
We use a duality theory approach to the assembly map as developed by Roe~\cite{Roe02,Roe04}. 
By the naturality of the long-exact sequence in $K$-theory (including van Daele $K$-theory), it suffices to 
consider the assembly map via the boundary map that arises from the short exact sequence
\[
  0 \to C^*(X) \to \mathrm{Mult}(C^*(X)) \to \mathrm{Mult}(C^*(X))/ C^*(X) \to 0
\]
rather than the short exact sequence from the ideal $C^*(X) \subset D^*(X)$~\cite[Proposition 5.11]{Roe96}. 
Let $\calQ(C^*(X)) = \mathrm{Mult}(C^*(X))/C^*(X)$ and $q: \mathrm{Mult}(C^*(X)) \to \calQ(C^*(X))$ the quotient map.
Proposition \ref{prop:Coarse_KHom_class_noSymm} implies that
$q( F_{\calC_{J_0}(J_1)})\otimes \rho \in \calQ(C^*(X)) \otimes \Cl_{0,1}$ is an odd self-adjoint unitary that anti-commutes 
with $q(J_0) \otimes \rho$. Let us now fix a reference 
odd self-adjoint unitary $e \in \calQ(C^*(X)) \otimes \Cl_{0,1}$ that lifts to an 
odd self-adjoint unitary in $\mathrm{Mult}(C^*(X))\otimes \Cl_{0,1}$ that 
 anti-commutes with $J_0 \otimes \rho$. Then 
$[F_{\calC_{J_0}(J_1)} \otimes \rho] \in DK_e( \calQ(C^*(X)) \otimes \Cl_{0,2})$, where the degree shift is because
$F_{\calC_{J_0}(J_1)}\otimes \rho$ and $e$ anti-commute with $J_0 \otimes \rho$ (cf. Remark~\ref{rk:extra_clifford_gens_shifts_the_degree}).
The coarse assembly map is given by the  composition
\[
  KO^{-2}(C_0(X)^\frakr) \to DK_e( \calQ(C^*(X)) \otimes \Cl_{0,2}) \xrightarrow{\partial} DK_e( C^*(X) \otimes \Cl_{0,1} ) \xrightarrow{\simeq} KO_2(C^*(X)^\frakr),
\]
where the first map is given by $\big[F_{\calC_{J_0}(J_1)}^\mathrm{Hom}\big] \mapsto \big[q(F_{\calC_{J_0}(J_1)})  \otimes \rho \big]$ and 
 $\partial$ is the boundary map in van Daele $K$-theory.
Applying Proposition \ref{prop:bdry_map_is_nice}, 
the boundary map 
composed with the equivalence between $DK$ and $KKR$ is represented by 
\begin{align*}
  \partial \big( \big[q(F_{\calC_{J_0}(J_1)})  \otimes \rho \big] \big) 
  &= \big[ ( \Cl_{1,0}, \, (C^*(X)\otimes \Cl_{0,1})_{C^*(X)\otimes \Cl_{0,1}},  \, F_{\calC_{J_0}(J_1)} \otimes \rho ) \big] \\
  &= \Big[ \big( \Cl_{2,0}, \, C^*(X)_{C^*(X)} \otimes \bigwedge\nolimits^{\! *} \C, \, F_{\calC_{J_0}(J_1)} \otimes \rho \big) \Big]
\end{align*}
with left Clifford generators $J_0\otimes \rho$ and $\{J_0\otimes \rho, \one \otimes \gamma\}$ in the 
first and second lines respectively.
Recalling Proposition \ref{prop:skew_cayley_kk} and Equation \eqref{eq:Bdd_transform_on_whole_module} 
on Page \pageref{eq:Bdd_transform_on_whole_module}, 
this Kasparov module is a representative of 
$\big[ \calC_{J_0}(J_1) \big]$ as required.
\end{proof}

Theorem \ref{thm:Coarse_Assembly_Works} shows that the coarse index associated to the pair of BdG Hamiltonians 
$H_0$ and $H_1$ satisfying Assumption \ref{assump:ControlledHamiltonians} encodes a topological obstruction to 
locally connect the skew-adjoint unitaries $J_0$ and $J_1$  with respect to $(X, \varphi)$, as is explained in the 
following corollary.

\begin{cor} \label{cor:Assembly_as_locality_obstruction}
Suppose that there is a continuous path $\{J_t\}_{t\in [0,1]}$ of Real skew-adjoint unitaries in $\mathrm{Mult}(C^*(X))$ such 
$\varphi(f)(J_0 -J_t) \in \mathbb{K}(\calH)$ for all $f \in C_0(X)$ and $t\in [0,1]$. Then the coarse 
index $\mu_X\big( [ F_{\calC_{J_0}(J_1)}^\mathrm{Hom} ] \big)$ is trivial in $KO_2(C^*(X)^\frakr)$.
\end{cor}
\begin{proof}
If such a path $\{J_t\}_{t\in [0,1]}$ exists, then the class $\big[ \calC_{J_0}(J_1)\big] \in KO_2(C^*(X)^\frakr)$ is trivial 
by Proposition \ref{prop:Connected_J_trivial_K}. Hence the coarse index vanishes.
\end{proof}

Let us also briefly consider the case of Altland--Zirnbauer symmetries, i.e., we assume that there 
exist mutually anti-commuting Real
skew-adjoint unitaries $\{\kappa_j\}_{j=1}^n \subset \mathrm{Mult}(C^*(X))$ that anti-commute with $J_0$ and $J_1$. 
We additionally assume that the 
eigenspaces of $\kappa_j$ are infinite dimensional for all $j \in\{1,\ldots, n\}$.
Such a circumstance trivially happens when 
$\calH = \calH' \otimes \calW$, where $\calH'$ is infinite-dimensional,  
$\varphi = \varphi' \otimes \one_\calW$ and $\kappa_j = \one_{\calH'} \otimes \kappa_j'$ with 
$\{\kappa_j'\}_{j=1}^n$ ungraded Clifford generators in $\End(\calW)$.

Proposition \ref{prop:Coarse_KHom_class_noSymm} along with our additional assumptions on $\{\kappa_j\}_{j=1}^n$ imply that
the triple 
\begin{equation} \label{eq:AZ_FredMod}
   \Big( C_0(X) \otimes \Cl_{n+2,0}, \, {}_\varphi \calH \otimes \bigwedge\nolimits^{\! *} \C, \, F_{\calC_{J_0}(J_1)} \otimes \rho \Big)
\end{equation}
is a Real Fredholm module with left Clifford generators $\{\one \otimes \gamma, J_0 \otimes \rho, \kappa_1 \otimes \rho, \ldots, \kappa_n \otimes \rho\}$.
We therefore obtain obtain a $K$-homology class $\big[ F_{\calC_{J_0}(J_1)}^\mathrm{Hom} \big] \in KO^{-n-2}(C_0(X)^\frakr)$.
We similarly have a $K$-theory class from the Kasparov module constructed in Remark \ref{rk:AZ_class_general} with $A=C^*(X)$, 
\begin{equation} \label{eq:AZ_Coarse_KasMod}
 \Big( \Cl_{n+2,0}, \, \ol{(J_0-J_1)C^*(X)}_{C^*(X)}  \otimes \bigwedge\nolimits^{\! *} \C , \, \calC_{J_0}(J_1) \otimes \rho \Big)
\end{equation}
with Clifford generators $\{\one \otimes \gamma, J_0 \otimes \rho, \kappa_1 \otimes \rho, \ldots, \kappa_n \otimes \rho\}$.
Like the case for $n=0$, our result is that the coarse assembly map relates the equivalence classes 
of the Kasparov modules in Equations \eqref{eq:AZ_FredMod}  
and \eqref{eq:AZ_Coarse_KasMod}.

\begin{thm}
Let $\big[ F_{\calC_{J_0}(J_1)}^\mathrm{Hom} \big] \in KO^{-n-2}(C_0(X)^\frakr)$ and $\big[ \calC_{J_0}(J_1) \big] \in KO_{n+2}(C^*(X)^\frakr)$ 
denote the classes from the Real Kasparov modules in Equations \eqref{eq:AZ_FredMod} and \eqref{eq:AZ_Coarse_KasMod} respectively. 
Then $\mu_X\big( [F_{\calC_{J_0}(J_1)}^\mathrm{Hom} ] \big) = \big[ \calC_{J_0}(J_1) \big]$.
\end{thm}
Taking into account the extra Clifford symmetries (cf. Remark \ref{rk:extra_clifford_gens_shifts_the_degree}), 
the proof follows the same argument as Theorem \ref{thm:Coarse_Assembly_Works}, where we 
 apply Proposition \ref{prop:bdry_map_is_nice} to 
the map $DK( \calQ(C^*(X)) \otimes \Cl_{0, n+1}) \xrightarrow{\partial\circ \calC} KKR( \Cl_{n+2,0}, C^*(X))$.

\subsection{Compact $G$-symmetry} \label{sec:Compact_G_assembly}

Let $G$ be a compact group and $W$ a unitary/anti-unitary representation on $\calH$ relative to 
a homomorphism $\nu: G\to \Z_2$, i.e., $W_g$ is unitary if $\nu(g)=0$ and anti-unitary 
if $\nu(g)=1$. We also assume that $W_g \Gamma = \Gamma W_g$. Let us similarly assume 
that there is a left-action $G\times X \to X$  
which gives rise to an linear/anti-linear 
action $\eta$ on $C_0(X)$ with respect to $\nu$. 
We therefore consider representations $\varphi: C_0(X) \to \calB(\calH)$ 
such that $\varphi\circ \eta_g( f ) = \mathrm{Ad}_{W_g} \circ \varphi(f)$ for all $f \in C_0(X)$ and $g \in G$.
Given  BdG Hamiltonians $H_0$ and $H_1$ satisfying Assumption \ref{assump:ControlledHamiltonians},
we assume that 
$W_g (\Dom(H_k) ) \subset \Dom(H_k)$ and $W_g H_k W_g^* = H_k$ for all $g \in G$ and $k\in \{0,1\}$. 
In particular, this 
implies that, for all $g \in G$,  $\mathrm{Ad}_{W_g}( J_k) = (-1)^{\nu(g)} J_k$ for $k \in \{0,1\}$ 
and $\mathrm{Ad}_{W_g} (\calC_{J_0}(J_1) ) = (-1)^{\nu(g)} \calC_{J_0}(J_1)$.

The action $\mathrm{Ad}_{W}$ on $\calB(\calH)$ gives a well-defined linear/anti-linear action on 
$C^*(X)$, $\mathrm{Mult}(C^*(X))$ and the quotient $\calQ(C^*(X)) = \mathrm{Mult}(C^*(X))/C^*(X)$
relative to the homomorphism $\nu: G\to \Z_2$ that commutes with the Real 
structure $\mathrm{Ad}_\Gamma$. 
We denote these actions on $C^*(X)$ and $\calQ(C^*(X))$ by $\beta^{C^*(X)}$  
and $\beta^{\calQ^*(X)}$ respectively.
We can then define the action $\wt{\alpha}$ of $G$ on the Hilbert $C^*$-module $C^*(X)_{C^*(X)} \otimes \bigwedge^* \C$
\begin{equation} \label{eq:Equiv_action_for_assembly}
  \wt{\alpha}_g( T \otimes v) = \beta^{C^*(X)}_g(T) \otimes \gamma^{\nu(g)}v, 
  \qquad  g \in G, \,\, \, T \otimes v \in C^*(X)\otimes \bigwedge\nolimits^{\! *}\C.
\end{equation}
One then checks that the induced action $\alpha$ of $G$ on 
$\End_{C^*(X)}(C^*(X)) \otimes \End( \bigwedge^*\C)$  is such that for all $g \in G$,
\[
   \alpha_g ( J_0 \otimes \rho) = J_0 \otimes \rho, \qquad
  \qquad \alpha_g(F_{\calC_{J_0}(J_1)} \otimes \rho) = F_{\calC_{J_0}(J_1)} \otimes \rho, \qquad \qquad 
  \alpha_g( \one \otimes \gamma) = \one \otimes \gamma.
\]
Analogously, we can define a group action on $\calH \otimes \bigwedge^* \C$ via the 
unitary/anti-unitary representation $\wt{W}$ such that 
$\wt{W}_g = W_g \otimes \gamma^{\nu(g)}$ for all $g \in G$. The operators 
$F_{\calC_{J_0}(J_1)} \otimes \rho,  J_0\otimes \rho$ 
and $\one \otimes \gamma$ will then be  $G$-invariant under the induced action $\mathrm{Ad}_{\wt{W}}$.
With the preliminaries established, we can now state the result.

\begin{thm}
The triple 
\[
  \Big( \Cl_{2,0}, \, {}_{\varphi}\calH \otimes \bigwedge\nolimits^{\! *} \C, \, F_{\calC_{J_0}(J_1)} \otimes \rho \Big) 
\]
is a Real $G$-equivariant Fredholm module with $G$-action by $\wt{W}$ and 
Clifford generators $\{ \one \otimes \gamma, J_0 \otimes \rho\}$. Under the equivariant 
assembly map $\mu_X^G: KO^{-2}(C_0(X)^\frakr) \to KO_2^G(C^*(X))$, the class of this Real Fredholm 
module is mapped to the $K$-theory class represented by the unbounded $G$-equivariant 
Real Kasparov module from Proposition \ref{prop:Equiv_K_theory_class_general} with $A=C^*(X)$,
\[
  \Big( \Cl_{2,0}, \, \ol{(J_0-J_1)C^*(X)}_{C^*(X)} \otimes \bigwedge\nolimits^{\! *} \C, \, \calC_{J_0}(J_1) \otimes \rho \Big)
\]
with $G$-action by $\wt{\alpha}$ from Equation \eqref{eq:Equiv_action_for_assembly}
and Clifford generators $\{ \one \otimes \gamma, J_0 \otimes \rho\}$.
\end{thm}
\begin{proof}
Because all relevant operators are $G$-invariant, 
the same proof as Proposition \ref{prop:Coarse_KHom_class_noSymm} and 
Theorem \ref{thm:Coarse_Assembly_Works} gives the result.  
\end{proof}

\subsection{Discrete $\Upsilon$-symmetries} \label{sec:cocompact_assembly}

A full discussion of discrete symmetries and the $\Upsilon$-equivariant assembly map 
deserves a separate treatment and so we will only give a basic overview.
We use the notation $\Upsilon$ for a discrete group to distinguish the setting 
from the case of a compact group action. 
We will furthermore restrict ourselves to \emph{linear} group actions ($\nu(g) = 0$ for all $g \in \Upsilon$).

Fix a discrete group $\Upsilon$ and a proper, isometric and cocompact left-action $\Upsilon \times X \to X$ giving 
a Real action $\eta:G\to \Aut(C_0(X))$. 
We similarly take  a unitary representation 
$V: \Upsilon \to \calU(\calH)$ such that $[V_g, \Gamma]=0$ and 
$\varphi \circ \eta_g(f) = \mathrm{Ad}_{V_g} \circ \varphi(f)$ for all $g\in \Upsilon$ and 
$f\in C_0(X)$. We will furthermore assume that $\calH$ is a $\Upsilon$-adequate 
$X$-module in the sense of~\cite[Definition 5.13]{Roe96} (this condition can always be guaranteed).

Once again we take BdG Hamiltonians $H_0$ and $H_1$ acting on $\calH$ that satisfy 
Assumption \ref{assump:ControlledHamiltonians} and furthermore for all $g \in \Upsilon$ and $k\in\{0,1\}$,
\[
  V_g \cdot \Dom(H_k) \subset \Dom(H_k), \qquad \qquad V_g H_k V_g^* = H_k.
\]
Let $\mathrm{Mult}(C^*(X))^\Upsilon$ and $C^*(X)^\Upsilon$ denote the subalgebras of 
$\mathrm{Mult}(C^*(X))$ and $C^*(X)$ respectively consisting of elements that 
are fixed by $\mathrm{Ad}_{V_g}$ for all $g \in \Upsilon$. 
Then because $V$ is a unitary representation, 
\[
  J_0, \, J_1 \in \mathrm{Mult}(C^*(X))^\Upsilon, \qquad \qquad J_0 - J_1 \in C^*(X)^\Upsilon.
\]

In~\cite[Section 2]{Roe02}, Roe constructs a full right Hilbert $C^*_r(\Upsilon)$-module ${L^2_\Upsilon(X)}_{C^*_r(\Upsilon)}$ such that 
$C^*(X)^\Upsilon$ is  isomorphic to $\mathbb{K}_{C^*_r(\Upsilon)}(L^2_\Upsilon(X))$. That is, 
$C^*(X)^\Upsilon$ is Morita equivalent to $C^*_r(\Upsilon)$ and we obtain an invertible element 
\[
   \big[ ( C^*(X)^\Upsilon, \, {L^2_\Upsilon(X)}_{C^*_r(\Upsilon)}, \, 0 ) \big] \in KKR( C^*(X)^\Upsilon,  C^*_r(\Upsilon) ).
\]
We can therefore construct a class in the $K$-theory of $C^*_r(\Upsilon)$ by composing this Morita equivalence with 
our generic $K$-theory construction from Proposition \ref{prop:skew_cayley_kk},
\begin{align*} 
  &\Big( \Cl_{2,0}, \,  \ol{(J_0-J_1)C^*(X)^\Upsilon}_{C^*(X)^\Upsilon} \otimes \bigwedge\nolimits^{\! *} \C, \, 
  \calC_{J_0}(J_1) \otimes \rho \Big) 
  \otimes_{ C^*(X)^\Upsilon } \big( C^*(X)^\Upsilon, \, {L^2_\Upsilon(X)}_{C^*_r(\Upsilon)}, \, 0 \big)  \nonumber \\
  &\qquad = \Big( \Cl_{2,0}, \, \ol{ (J_0-J_1)C^*(X)^{\Upsilon}\cdot L^2_\Upsilon(X)}_{C^*_r(\Upsilon)} \otimes \bigwedge\nolimits^{\! *} \C, 
  \, \calC_{J_0}(J_1) \otimes \rho \Big)
\end{align*}
with left Clifford generators $\{J_0 \otimes \rho, \one \otimes \gamma\}$.
Noting that $J_0, \, F_{\calC_{J_0}(J_1)} \in \mathrm{Mult}(C^*(X))^\Upsilon  \subset  \End_{C^*_r(\Upsilon)}( L^2_\Upsilon(X) )$ and 
using the properties of the bounded transform $F_{\calC_{J_0}(J_1)}$ from Equations \eqref{eq:Skew_Cayley_F} 
and \eqref{eq:Skew_Cayley_F_properties} on Page \pageref{eq:Skew_Cayley_F}, 
a bounded representative of this Kasparov module is given by 
\begin{equation} \label{eq:Gp_K-theory_class}
   \Big( \Cl_{2,0}, \, L^2_\Upsilon(X)_{C^*_r(\Upsilon)} \otimes \bigwedge\nolimits^{\! *} \C, \, F_{\calC_{J_0}(J_1)} \otimes \rho \Big)
\end{equation}
with left Clifford generators $\{J_0 \otimes \rho, \one \otimes \gamma\}$.

We can define a representation of $\Upsilon$ on $\calH \otimes \bigwedge^*\C$ by 
$\wt{V}_g( v \otimes w ) = V_g v \otimes w $ for all $v \otimes w \in \calH \otimes \bigwedge^*\C$ and $g\in \Upsilon$. 
Because $F_{\calC_{J_0}(J_1)}$ and $J_0$ are invariant under $\mathrm{Ad}_{V_g}$ for all $g\in \Upsilon$, 
Proposition \ref{prop:Coarse_KHom_class_noSymm} once again gives that 
\begin{equation} \label{eq:DiscGp_FredMod}
  \Big( C_0(X) \otimes \Cl_{2,0}, \, {}_{\varphi}\calH \otimes \bigwedge\nolimits^{\! *} \C, \, F_{\calC_{J_0}(J_1)} \otimes \rho \Big)
\end{equation}
is a $\Upsilon$-equivariant Real Fredholm module with group action by $\wt{V}$ and 
left Clifford generators $\{\one \otimes \gamma, J_0 \otimes \rho\}$.

\begin{thm}
The coarse assembly map $\mu_X^\Upsilon: KO^{-2}_\Upsilon( C_0(X)^\rs) \to KO_2( C^*_r( \Upsilon)^\frakr )$ 
maps the class of the Real Fredholm module from Equation \eqref{eq:DiscGp_FredMod} to the $K$-theory class 
represented by the Real Kasparov module from Equation \eqref{eq:Gp_K-theory_class}.
\end{thm}
\begin{proof}
Once again, we use Roe's  description of the coarse assembly map via duality theory~\cite{Roe02}. 
The coarse assembly map is the composition
\[
  KO^{-2}_\Upsilon(C_0(X) ) \to DK_e( \calQ(C^*(X))^\Upsilon \otimes \Cl_{0,2} ) \xrightarrow{\partial} DK_e( C^*(X)^\Upsilon \otimes \Cl_{0,1} ) 
  \xrightarrow{\simeq} KO_2( C^*_r( \Upsilon)^\frakr ),
\]
with $\calQ(C^*(X))^\Upsilon = \mathrm{Mult}(C^*(X))^\Upsilon / C^*(X)^\Upsilon$ and 
in the last step we use the Morita equivalence of $C^*(X)^\Upsilon$ with $C^*_r(\Upsilon)$. 
The class of the Fredholm module from Equation \eqref{eq:DiscGp_FredMod}  is initially mapped to 
$[q(F_{\calC_{J_0}(J_1)}) \otimes \rho]$. Then 
by the same argument as  Theorem \ref{thm:Coarse_Assembly_Works},
$\partial [q(F_{\calC_{J_0}(J_1)}) \otimes \rho]$ combined with the 
 isomorphism $DK( C^*(X)^\Upsilon \otimes \Cl_{0,1} ) \cong KKR( \Cl_{2,0}, C^*(X)^\Upsilon )$ is 
represented by the Kasparov module 
$\big( \Cl_{2,0},  \, {C^*(X)}^\Upsilon \otimes \bigwedge^* \C, \, F_{\calC_{J_0}(J_1)} \otimes \rho \big)$ 
with $\Cl_{2,0}$-generators $\{\one \otimes \gamma, J_0 \otimes \rho\}$. 
 Applying 
the Morita equivalence, 
\begin{align*}
  &\Big( \Cl_{2,0}, \,  C^*(X)^{\Upsilon}_{C^*(X)^\Upsilon} \otimes \bigwedge\nolimits^{\! *} \C, \,  F_{\calC_{J_0}(J_1)} \otimes \rho \Big) 
  \otimes_{ C^*(X)^\Upsilon } \big( C^*(X)^\Upsilon, \, {L^2_\Upsilon(X)}_{C^*_r(\Upsilon)}, \, 0 \big)  \nonumber \\
  &\qquad =   \Big( \Cl_{2,0}, \, L^2_\Upsilon(X)_{C^*_r(\Upsilon)} \otimes \bigwedge\nolimits^{\! *} \C, \, F_{\calC_{J_0}(J_1)} \otimes \rho \Big)
\end{align*}
with $\Cl_{2,0}$-generators $\{\one \otimes \gamma, J_0 \otimes \rho\}$. 
We therefore obtain 
the Real Kasparov module from Equation \eqref{eq:Gp_K-theory_class}.  
\end{proof}

\begin{remarks}
\begin{enumerate}
  \item[(i)]  The $\Upsilon$-equivariant assembly map, 
  interpreted as a higher index, provides a topological 
obstruction to the existence of a  $\Upsilon$-invariant path of Real skew-adjoint unitaries 
$\{J_t\}_{t\in [0,1]}\subset \mathrm{Mult}(C^*(X))^\Upsilon$ such that  
$\varphi(f)(J_0-J_t) \in \mathbb{K}(\calH)$ for all $f \in C_0(X)$ 
and $t \in [0,1]$. As in the non-equivariant setting (Corollary \ref{cor:Assembly_as_locality_obstruction}), 
the existence of such a path implies that $\mu_X^\Upsilon$ applied to the $K$-homology class 
from Equation \eqref{eq:DiscGp_FredMod}  
will be trivial in $KO_2(C^*_r(\Upsilon))$. 

  \item[(ii)] In the case that $\Upsilon = S$, a space group embedded in Euclidean space, the group 
$KO_\ast( C^*_r(S) )$ has been extensively studied as way to classify free-fermionic topological insulators and 
superconductors, see~\cite{Kubota16, GKT21} for example. A more comprehensive comparison between 
the $K$-homology invariants that arise from locally equivalent quasifree ground states with the $K$-theory invariants 
from free-fermionic topological phases would be interesting to consider. We leave this question to another place.
\end{enumerate}

\end{remarks}

\subsection*{Acknowledgements}
Thanks to Ken Shiozaki, who first explained to me that free-fermionic SPT phases should be 
classified by $K$-homology. I also thank Yosuke Kubota for many helpful discussions on this topic. 
Finally, I thank the anonymous referees whose careful reading and detailed comments have 
greatly improved the paper.
This work is supported by a JSPS Grant-in-Aid for Early-Career Scientists (No. 19K14548).

\end{document}